\documentclass[mnsc]{informs3}

\OneAndAHalfSpacedXI 

\TheoremsNumberedThrough     
\ECRepeatTheorems

\EquationsNumberedThrough    

\renewenvironment{proof}[1][Proof]{\textsc{#1.} }{ \hfill $\square$ \linebreak}
\let\argmin\relax
\let\argmax\relax
\renewcommand\theARTICLETOP{}       
\RRHSecondLine{}					
\LRHSecondLine{}					
\usepackage[x11names]{xcolor}
\usepackage[colorlinks]{hyperref}
\hypersetup{citecolor=SpringGreen4}

\usepackage{thmtools,enumitem}

\newenvironment{myexample}{\begin{example}}{%
    \end{example}\medskip
}

\usepackage{amsmath,amsfonts,dsfont,amssymb,comment}
\usepackage{multirow}

\newcommand{\R}{\mathbb{R}}								
\newcommand{\N}{\mathbb{N}}								
\newcommand{\Rp}{\mathbb{R}_{\geq 0}}					
\newcommand{\E}{\mathbb{E}}								
\renewcommand{\Pr}{\mathbb{P}}							
\newcommand{\In}[1]{\mathds{1}_{\crl*{ #1}}} 			
\newcommand{\Ins}[1]{\mathds{1}_{#1}} 			

\newcommand{\defeq}{\coloneqq}							

\usepackage{mathtools}
\DeclarePairedDelimiter{\abs}{\lvert}{\rvert}
\DeclarePairedDelimiter{\norm}{\lvert\lvert}{\rvert\rvert}
\DeclarePairedDelimiter{\crl}{\{}{\}} 
\DeclarePairedDelimiter{\prn}{(}{)} 
\DeclarePairedDelimiter{\brk}{[}{]} 

\newcommand{\T}{T}                                      
\usepackage{xspace}
\newcommand{\voff}{V^{\textnormal{off}}}						
\newcommand{\von}{V^{\textnormal{on}}}						
\newcommand{\rmax}{r_{\max}}							
\newcommand{\pmin}{p_{\min}}                            
\newcommand{\Ron}{R^{\text{on}}}						
\newcommand{\Jt}{\theta^t}                                   
\newcommand{\JT}[1]{\theta^{#1}}                                   
\newcommand{\onl}{\textsc{Online}\xspace}
\newcommand{\off}{\textsc{Offline}\xspace}
\newcommand{\reg}{\textsc{Reg}\xspace}
\DeclareMathOperator{\Bin}{Bin}
\DeclareMathOperator{\Mult}{Multinomial}
\DeclareMathOperator{\Poiss}{Poisson}
\DeclareMathOperator{\Ber}{Bernoulli}

\DeclareMathOperator{\argmin}{argmin}
\DeclareMathOperator{\argmax}{argmax}

\newcommand{\Xt}{X^t}									
\newcommand{\Xts}{X^{\star t}}							
\newcommand{\pon}{\pi^{\text{on}}}						
\newcommand{\poff}{\pi^{\text{off}}}	
\newcommand{\rl}{\partial r}								
\newcommand{\Rl}{\partial R}                                
\newcommand{\Ac}{\mathcal{A}}							
\newcommand{\St}{\mathcal{S}}							
\newcommand{\J}{\Theta}
\newcommand{\Tr}{\mathcal{T}}							
\renewcommand{\Re}{\mathcal{R}}                         
\newcommand{\vx}{\mathbf{x}}							
\newcommand{\calF}{\mathcal{F}}                         
\newcommand{\jstar}{{j^\star}}

\newcommand{\pb}{\bar p}                                

\renewcommand{\paragraph}[1]{\noindent\textbf{#1}}

\usepackage{algorithm}
\usepackage[noend]{algpseudocode}

\newenvironment{proofof}[1]{%
	\begin{proof}[{\sc Proof of #1}]%
	}{%
	\end{proof}%
}

\usepackage{epigraph}
\usepackage[capitalize]{cleveref}

\usepackage[square,sort,comma,numbers]{natbib}
\begin{document}


 \RUNAUTHOR{Alberto Vera and Siddhartha Banerjee}

\RUNTITLE{The Bayesian Prophet: A Low-Regret Framework for Online Decision Making}

\TITLE{The Bayesian Prophet: A Low-Regret Framework for Online Decision Making}

\ARTICLEAUTHORS{%
\AUTHOR{Alberto Vera, Siddhartha Banerjee}
\AFF{Cornell University, \EMAIL{aav39@cornell.edu}, \EMAIL{sbanerjee@cornell.edu}}
} 

\ABSTRACT{%
We develop a new framework for designing online policies given access to an oracle providing statistical information about an offline benchmark. 
Having access to such prediction oracles enables simple and natural Bayesian selection policies, and raises the question as to how these policies perform in different settings. 
Our work makes two important contributions towards this question: 
First, we develop a general technique we call \emph{compensated coupling} which can be used to derive bounds on the expected regret (i.e., additive loss with respect to a benchmark) for any online policy and offline benchmark. 
Second, using this technique, we show that a natural greedy policy, which we call \emph{the Bayes Selector}, has constant expected regret (i.e., independent of the number of arrivals and resource levels) for a large class of problems we refer to as Online Allocation with finite types, which includes widely-studied Online Packing and Online Matching problems. 
Our results generalize and simplify several existing results for Online Packing and Online Matching, and suggest a promising pathway for obtaining oracle-driven policies for other online decision-making settings.
}%


\maketitle

\epigraph{\textit{``Life is the sum of all your choices."}}{--- \textup{Albert Camus}}

\section{Introduction}

\label{sec:intro}
Everyday life is replete with settings where we have to make decisions while facing uncertainty over future outcomes. 
Some examples include allocating cloud resources, matching an empty car to a ride-sharing passenger, displaying online ads, selling airline seats, etc.
In many of these instances, the underlying arrivals arise from some known generative process. 
Even when the underlying model is unknown, companies can turn to ever-improving machine learning tools to build predictive models based on past data.
This raises a fundamental question in online decision-making: \emph{how can we use predictive models to make good decisions?} 

Broadly speaking, an online decision-making problem is defined by a current state and a set of actions, which together determine the next state as well as generate rewards.
In Markov decision processes (MDPs), the rewards and state transitions are also affected by some random shock.
Optimal policies for such problems are known only in some special cases, when the underlying problem is sufficiently simple, and knowledge of the generative model sufficiently detailed. 
For many problems of interest, an MDP approach is infeasible due to two reasons: $(1)$ insufficiently detailed models of the generative process of the randomness, and $(2)$ the complexity of computing the optimal policy (the so-called `curse of dimensionality'). 
These shortcomings have inspired a long line of work on approximate dynamic programming (ADP). 

We focus on a general class of online resource allocation problems, which we refer to as Online Allocation (cf. Section~\ref{sec:def_allocation}), and which generalizes two important classes of online decision-making problems: Online Packing and Online Matching.
In brief, Online Allocation problems involve a set of $d$ distinct resources, and a principal with some initial budget vector $B\in\N^d$ of these resources, which have to be allocated among $T$ incoming agents. 
Each agent has a type comprising of some specific requirements for resources and associated rewards. 
The exact type becomes known only when the agent arrives. 
The principal must make irrevocable accept/reject decisions to try and maximize rewards, while obeying the budget constraints.

Online Packing and Online Matching problems are fundamental in MDP theory; they have a rich existing literature and widespread applications in many domains.
Nevertheless, our work develops new policies for both these problems which admit performance guarantees that are order-wise better than existing approaches. 
These policies can be stated in classical ADP terms (for example, see~\cref{alg:fluid,alg:fluid_matching}), but draw inspiration from ideas in Bayesian learning. 
In particular, our policies can be derived from a meta-algorithm, the Bayes selector (\cref{alg:bayes_selector}), which makes use of a black-box \emph{prediction oracle} to obtain statistical information about a chosen offline benchmark, and then acts on this information to make decisions.
Such policies are simple to define and implement in practice, and our work provides new tools for bounding their \emph{regret} vis-\'{a}-vis the offline benchmark. 
Thus, we believe that though our theoretical guarantees focus on a particular class of Online Allocation problems, our approach provides a new way for designing and analyzing much more general online decision-making policies using predictive models.

\subsection{Our Contributions}

We believe our contributions in this work are threefold:
\begin{enumerate}
\item \emph{Technical}: We present a \emph{new stochastic coupling technique}, which we call the \emph{compensated coupling}, for evaluating the regret of online decision-making policies vis-\`{a}-vis offline benchmarks.
\item \emph{Methodological}: Inspired by ideas from Bayesian learning, we propose a class of policies, expressed as the \emph{Bayes Selector}, for general online decision-making problems.
\item \emph{Algorithmic}: For a wide class of problems which we refer to as Online Allocation (which includes Online Packing and Online Matching problems), we prove that the Bayes Selector gives expected regret guarantees that are \emph{independent of the size of the state-space}, i.e., constant with respect to the horizon length and budgets. 
\end{enumerate}

\paragraph{Organization of the paper:} 
The rest of the paper is organized as follows: In~\cref{sec:prelim}, we introduce a general problem, called \emph{Online Allocation}, which includes as special cases the Multi-Secretary, Online Packing and Online Matching problems, and also more general settings involving agents with complex valuations over bundles. We also define the prophet benchmark, and discuss shortcomings of prevailing approaches. 
In~\cref{sec:coupling}, we present our main technical tool, the Compensated Coupling, in the general context of finite-state finite-horizon MDPs; we illustrate its use by applying it to the ski-rental problem.
In \cref{sec:bayes_selector} we introduce the Bayes Selector policy, and discuss how the compensated coupling provides a generic recipe for obtaining regret bounds for such a policy. 
In \cref{sec:regret_rm,sec:matching_type}, we use these techniques for the Online Packing and Online Matching problems; we analyze them separately to exploit their structure and obtain stronger results.
Finally, in \cref{sec:general_allocation} we analyze the most general problem (Online Allocation).

In particular, in~\cref{sec:regret_rm}, we propose a Bayes Selector policy for Online Packing and demonstrate the following performance guarantee:
\begin{theorem}[Informal]
For any Online Packing problem with a finite number of resource types and arrival types, under mild conditions on the arrival process, the regret of the Bayes Selector is independent of the horizon $T$ and budgets $B$ (in expectation and with high probability). 
\end{theorem}
In more detail, our regret bounds depend on the `resource matrix' $A$ and the distribution of arriving types, but are independent of $T$ and $B$.
Moreover, the results holds under weak assumptions on the arrival process, including Multinomial and Poisson arrivals, time-dependent processes, and Markovian arrivals. 
This result generalizes prior and contemporaneous results~\citep{Reiman_nrm,jasin2012,wang_resolve,wu2015algorithms,itai_secretary}.
We show similar results for Online Matching problems in \cref{sec:matching_type}. 

\subsection{Related Work}\label{sec:related}

Our work is related to several active areas of research in MDPs and online algorithms.

\paragraph{Approximate Dynamic Programming:}
The complexity of computing optimal MDP solutions can scale with the state space, which often makes it impractical (the so-called `curse of dimensionality'~\cite{powell2011approximate}). This has inspired a long line of work on \emph{approximate dynamic programming} (ADP)~\cite{powell2011approximate,tsitsiklis2001regression} to develop lower complexity heuristics. 
Although these methods often work well in practice, they require careful choice of basis functions, and any bounds are usually in terms of quantities which are difficult to interpret. Our work provides an alternate framework, which is simpler and has interpretable guarantees.

\paragraph{Model Predictive Control:} 
A popular heuristic for ADP and control which is closer to our paradigm is that of \emph{model predictive control} (MPC)~\cite{morari1993model,borrelli2003constrained,ciocan2012model}. MPC techniques have also been connected with online convex optimization (OCO) ~\cite{huang2015receding,chen2016using,chen2015online} to show how prediction oracles can be used for OCO, and applying these policies to problems in power systems and network control. These techniques however require continuous controls, and do not handle combinatorial constraints.

\paragraph{Information Relaxation:}
Parallel to the ADP focus on developing better heuristics, there is a line of work on deriving upper bounds via martingale duality, sometimes referred to as information relaxations~\cite{brown2013optimal,desai2012pathwise,brown2014information}. 
The main idea in these works is to obtain performance bounds for heuristic policies work by defining more refined outer bounds; in particular, this can be done by adding a suitable martingale term to the current reward, in order to penalize `future information'. 
Our offline benchmarks serve a similar purpose; however, a critical difference is that instead of using these to analyze a given heuristic, we use the benchmarks directly to derive control policies.

\paragraph{Online Packing and Prophet Inequalities:}
The majority of work focuses on competitive ratio bounds under worst-case distributions. In particular, there is an extensive literature on the so-called Prophet Inequalities, starting with the pioneering work of \citep{hill_iid}, to more recent extensions and applications to algorithmic economics~\citep{kleinberg2012matroid,duetting2017prophet,correa2017posted,alaei_bayesian}.
We note that any competitive ratio guarantee implies a $O(T)$ expected regret, in comparison to our $O(1)$ expected regret guarantees -- the cost for this, however, is that our results hold under more restrictive assumptions on the inputs.
For example, the policy suggested by \citep{duetting2017prophet} is static and \citep[Theorem 1]{itai_secretary} shows that any static policy has $\Omega(\sqrt{T})$ expected regret, hence it cannot yield a strong guarantee like ours.

\paragraph{Distribution-agnostic and Adversarial Models:}
Though we focus only on the case where the input is drawn from a stochastic process, we note that there is a long line of work on Online Packing with adversarial inputs~\citep{buchbinder2009online,buchbinder2009design,kesselheim2018primal}, and also distribution-agnostic approaches~\citep{badanidiyuru2013bandits,devanur2019near}.
More generally, there is a large body of work on using sublinear expected regret algorithms for solving online linear and convex programs
\citep{agrawal2014fast,gupta2014experts}. The algorithms in these works are incomparable to ours in that, while they cannot get constant expected regret in our setting (stochastic input, finite type space), they provide guarantees under much weaker assumptions.

\paragraph{Regret bounds in Stochastic Online Packing:}
For these problems, regret is the most meaningful metric to study, see \citep{cong_shi} for a discussion, where approximations to the regret are studied.
The first work to prove constant expected regret in a context similar to ours is \citep{itai_secretary}, who prove a similar result for the multi-secretary setting with multinomial arrivals; we strengthen their result in \cref{theo:secretary}.
This result is relevant to a long line of work in applied probability.
Some influential works are \citep{Reiman_nrm}, which provides an asymptotically optimal policy under the diffusion scaling, and~\citep{jasin2012} who provide a resolving policy with constant expected regret under a certain non-degeneracy condition. In contrast, degeneracy plays no role in the performance of our algorithms.
More recently, \citep{wang_resolve} partially extended the result of \citep{itai_secretary} for more general packing problems; their guarantee is only valid for i.i.d. Poisson arrival processes and when the system is scaled linearly, i.e., when $B$ is proportional to $T$ (our results and \citep{itai_secretary} make no such assumption).
In \cref{sec:numerics}, we demonstrate with a numerical study that the Bayes Selector far outperforms all these previous policies.

\section{Problem Setting and Overview of Results}
\label{sec:prelim}

As we mentioned in the introduction, our contributions in this work are two-fold -- (i) we give a technique to analyze the regret of any MDP, and (ii) we apply it to specific problems to obtain constant regret. 
Our focus in this work is on the subclass of \emph{Online Packing problems with stochastic inputs}. 
This is a subclass of the wider class of \emph{finite-horizon online decision-making problems}: given a time horizon $\T\in\N$ with discrete time-slots $t = \T,\T-1,\ldots,1$, we need to make a decision at each time leading to some cumulative reward. 
Note that throughout our time index $t$ indicates the \emph{time to go}.
We present the details of our technical approach in this more general context whenever possible, indicating additional assumptions when required.

In what follows, we use $[k]$ to indicate the set $\crl{1,2,\ldots,k}$, and denote the $(i,j)$-th entry of any given matrix $A$ interchangeably by $A_{i,j}$ or $A(i,j)$.
We work in an underlying probability space $(\Omega,\calF,\Pr)$, and the complement of any event $Q\subseteq\Omega$ is denoted $\bar Q$.
For any optimization problem $(P)$, we use $v(P)$ to indicate its objective value.
If $S$ is a finite set, $\abs{S}$ denotes cardinality. 
The set $\N$ of naturals includes zero.

\subsection{The Online Allocation Problem}\label{sec:def_allocation}

We now present a generic problem, which we refer to as \emph{Online Allocation}, that encompasses both Online Matching and Online Packing.
The setup is as follows: There are $d$ distinct resource-types denoted by the set $[d]$, and at time $t=\T$, we have an initial availability (budget) vector $B = (B_1,B_2,\ldots,B_d)\in\N^d$.
At every time $t=T,T-1,\ldots,1$, an arrival with \emph{type} $\Jt$ is drawn from a finite set of $n$ distinct types $\J = [n]$, via some distribution which is known to the algorithm designer (henceforth referred to as the \emph{principal}). 
We denote $Z(t) = (Z_1(t),Z_2(t),\ldots,Z_n(t))\in \N^n$ as the cumulative vector of the last $t$ arrivals, where $Z_j(t)\defeq\sum_{\tau\leq t}\In{\JT{\tau}=j}$.

Each arriving agent is associated with a \emph{choice-over-bundles}. Formally, an arriving agent of type $j$ desires any one among a collection of multisets $S_j\subseteq 2^{[d]}$; the principal can allocate any $s\in S_j$ (referred to as a \emph{bundle}) to the agent, thereby obtaining a reward $r_{sj}$ and consuming one unit of each resource $i\in s$.
Observe that we do not assume additive valuations, e.g, we do not require $r_{\crl{1,2}j} = r_{\crl{1}j}+r_{\crl{2}j}$. The model can naturally be extended to allow bundles to consume multiple units of any resource. 

At each time, the principal must decide whether to allocate a bundle to the request $\Jt$ (thereby generating the associated reward while consuming the required resources), or reject it (no reward and no resource consumption). 
Allocating a bundle requires that there is sufficient budget of each resource to cover the request. 
The principal's aim is to make irrevocable decisions so as to maximize overall rewards.

In \cref{sec:general_allocation}, we present results for this general problem.
We next describe three particular cases of the general problem, which are each of independent interest.
We analyze these three cases separately since they admit improved results over the general case.

\paragraph{Multi-Secretary.} This is a fundamental one dimensional instance.
In this problem, we have $B\in\N$ available positions and want to hire employees with the highest abilities (rewards).
There is one resource type $(d=1)$ with budget $B$; each employee occupies one unit of budget (one position).
This is an instance of Online Allocation, with $S_j=\crl{\crl{1}}$ for all $j$ (all candidates want the same resource) and rewards $r_{\crl{1}j}=r_j$.

\paragraph{Online Packing.} In this multidimensional problem, each request $j$ is associated to one bundle and one reward if allocated.
Specifically, we are given a consumption matrix $A\in\N^{d\times n}$, where $a_{ij}$ denotes the units of resource $i$ required to serve the request $j$ and, for each $j$, there is a reward $r_j$ if served.
This is an instance of Online Allocation, wherein agents of each type $j$ desire a single bundle: $S_j=\crl{\crl{a_{ij} \text{ units of } i\text{ for each resource } i\in[d]}}$.

\paragraph{Online Matching.} There are $d$ resources, but now each type $j$ wants \emph{at most one resource} from among a given set of resources. Formally, a type $j$ agent wants any from $A_j$ instead of all from $A_j$.
The types can be represented by a reward matrix $r\in \Rp^{d\times n}$ and adjacency matrix $A\in \crl{0,1}^{d\times n}$; if the arrival is of type $j\in [n]$, we can allocate at most one resource $i$ such that $a_{ij}=1$, leading to a reward of $r_{ij}$.
This problem can be thought as online bipartite matching, see \cref{sec:matching_type} for details.
It corresponds to an instance of Online Allocation, with bundles $S_j=\crl{\crl{i} \text{ for each resource } i\in [d] \text{ s.t. } a_{ij}=1}$ and rewards $r_{\crl{i}j}=r_{ij}$.

\subsection{Arrival Processes} 
To specify the generative model for the type sequence $\JT{T},\JT{T-1},\ldots,\JT{1}$, an important subclass is that of \emph{stationary independent} arrivals, which further admits two widely-studied cases:\\
\noindent $1.$ The Multinomial process is defined by a known distribution $p\in\Rp^n$ over $[n]$; at each time, the arrival is of type $j$ with probability $p_j$, thus $Z(t)\sim \Mult(t,p_1,\ldots,p_n)$.\\
\noindent $2.$ The Poisson arrival process is characterized by a known rate vector $\lambda\in\Rp^n$. Arrivals of each class are assumed to be independent such that $Z_j(t)\sim \Poiss(\lambda_jt)$. 
Note that, although this is a continuous-time process, it can be accommodated in a discrete-time formulation by defining as many periods as arrivals  (see \cref{sec:poisson} for details).

We assume w.l.o.g.\ that $p_j>0$ and $\lambda_j>0$ for all $j\in [n]$ (if this is not the case for some $j$, that type never arrives and can be removed from the instance description).
More general models allow for non-stationary and/or correlated arrival processes -- for example, non-homogeneous Poisson processes, Markovian models (see \cref{ex:markov}), etc. 
An important feature of our framework is that it is capable of handling a wide variety of such processes in a unified manner, without requiring extensive information regarding the generative model. 
We discuss the most general assumptions we make on the arrival process in Section~\ref{sec:warm_up2}.

\subsection{The Offline Benchmark} 
Suppose a given problem is simultaneously solved by two `\emph{agents}', \onl and \off, who are primarily differentiated based on their access to information. 
\onl can only take \emph{non-anticipatory} actions, i.e., use available information only, whereas \off is allowed to make decisions with full  knowledge of future arrivals.
This is known in the literature as a \emph{prophet or full-information benchmark}.
Denoting the total collected rewards of \off and \onl  as $\voff$ and $\von$ respectively, we define the \emph{regret} to be the \emph{additive} loss $\reg\defeq\voff-\von$.
Observe that $\von$ depends on the policy used by \onl, the underlying policy will always be clear from context. 
Our aim is to design policies with low $\E[\reg]$.

For Online Packing, the solution to \off's problem corresponds to solving an integer programming problem.
A looser, but more tractable benchmark, is given by an LP relaxation of this policy: given arrivals vector $Z(\T)$, we assume \off solves the following:
\begin{equation}\label{eq:off_problem}
\begin{array}{rrl}
P[Z(\T),B]:\quad \max  & r'x& \\
\text{s.t.}&  Ax&\leq\,\, B \\
& x&\leq\,\, Z(\T) \\
& x&\geq\,\, 0.
\end{array}
\end{equation}

\paragraph{The unavoidable regret of the fluid benchmark}:
The most common technique for obtaining policies for Online Packing is based on the so-called fluid (a.k.a.\ deterministic or ex ante) LP benchmark $(P[\E[Z(\T)],B])$, where $(P)$ is defined in \cref{eq:off_problem}. 
It is easy to see via Jensen's Inequality that $v(P[\E[Z(\T)],B]) \geq \E[v(P[Z(\T),B])]$, and hence the fluid LP is an upper bound for any online policy. 
Although the use of this fluid benchmark is the prevalent tool to bound the regret in Online Packing problems \cite{talluri2006theory,Reiman_nrm,jasin2012,wu2015algorithms}, the following result shows that the approach of using $v(P[\E[Z(\T)],B])$ as a benchmark can never lead to a constant expected regret policy, as the fluid benchmark can be far off from the optimal solution in hindsight. 
\begin{proposition}
\label[proposition]{prop:bad_fluid}
For any Online Packing problem, if the arrival process satisfies the Central Limit Theorem and the fluid LP is dual degenerate, i.e., the optimal dual variables are not unique, then $v(P[\E[Z(\T)],B]) - \E[v(P[Z(\T),B])] = \Omega(\sqrt{\T})$.
\end{proposition}
This gap has been reported in literature, both informally and formally (see~\citep{itai_secretary,wang_resolve}).
 For completeness we provide a proof in~\cref{sec:proofs}.
Note though that this gap does not pose a barrier to showing constant-factor competitive ratio guarantees, i.e. $O(T)$ expected regret; the fluid LP benchmark is widely used for prophet inequalities.
In contrast, the gap presents a barrier for obtaining $O(1)$ expected regret bounds. 
Breaking this barrier thus requires a fundamentally new approach.

\subsection{Overview of our Approach and Results}
Our approach can be viewed as a meta-algorithm that uses black-box prediction oracles to make decisions.
The quantities estimated by the oracles are related to our offline benchmark and can be interpreted as {probabilities of regretting each action in hindsight}.
A natural `Bayesian selection' strategy given such estimators is to \emph{adopt the action that is least likely to cause regret in hindsight}.
This is precisely what we do in~\cref{alg:bayes_selector}, and hence, we refer to it as the Bayes Selector.

Bayesian selection techniques are often used as heuristics in practice. 
Our work however shows that such policies in fact have excellent performance for online allocation; in particular, we show that for matching and packing problems:
\begin{enumerate}
    \item There are easy to compute estimators (in particular, ones which are based on simple adaptive LP relaxations) that, when used for \cref{alg:bayes_selector}, give constant expected regret for a wide range of distributions (see \cref{theo:secretary,theo:reg_general,theo:matching}).
    \item Using other types of estimators (for example, Monte Carlo estimates) in \cref{alg:bayes_selector} yields comparable performance guarantees (see  \cref{cor:general,cor:matching}). 
\end{enumerate}

At the core of our analysis is a \emph{novel stochastic coupling technique} for analyzing online policies based on offline (or \emph{prophet}) benchmarks. 
Unlike traditional approaches for regret analysis which try to show that an online policy tracks a fixed offline policy, our approach is instead based on \emph{forcing \off to follow \onl's actions}. 
We describe this in more detail in the next section.

\section{Compensated Coupling and the Bayes Selector}
\label{sec:coupling}

We introduce our two main technical ideas: $1.$ the \emph{compensated coupling} technique, and $2.$ the \emph{Bayes selector} heuristic for online decision-making. 
The techniques introduced here are valid for any generic MDP; in subsequent sections we specialize them to Online Allocation.

\subsection{MDPs and Offline Benchmarks}
\label{ssec:mdp}

The basic MDP setup is as follows: at each time $t=T,T-1,\ldots,1$ (where $t$ represents the time-to-go), based on previous decisions, the system \emph{state} is one of a set of possible states $\St$.
Next, nature generates an \emph{arrival} $\Jt\in\J$, following which we need to choose from a set of available \emph{actions} $\Ac$. 
The state updates and rewards are determined via a transition function $\Tr:\Ac\times \St\times\J \to \St$ and a reward function $\Re:\Ac\times \St\times\J\to \R$: for current state $s\in\St$, arrival $j\in\J$ and action $a\in\Ac$, we transition to the state $\Tr(a,s,j)$ and collect a reward $\Re(a,s,j)$.
Infeasible actions $a$ for a given state $s$ correspond to $\Re(a,s,j)=-\infty$.
The sets $\Ac,\St,\J$, as well as the measure over arrival process $\crl{\Jt:t\in[\T]}$, are known in advance.
Finally, though we focus mainly on maximizing rewards, the formalism naturally ports over to cost-minimization.

Recall that we adopt the view that the problem is simultaneously solved by two `\emph{agents}': \onl and \off.
\onl can only take \emph{non-anticipatory} actions while \off makes decisions with knowledge of future arrivals. 
To keep the notation simple, we restrict ourselves to deterministic policies for \off and \onl, thereby implying that the only source of randomness is due to the arrival process (our results can be extended to randomized policies). 

A sample path $\omega\in \Omega$ encodes the arrival sequence $\crl{\Jt:t\in[\T]}$.
In other words, there exists a unique sequence of types that is consistent with $\omega$; whenever we fix $\omega$, the type $\Jt$ is uniquely identified, but, for notational ease, we do not write $\Jt[\omega]$.
For a given sample-path $\omega\in\Omega$ and \emph{time $t$ to go},
\off's value function is specified via the \emph{deterministic} Bellman equations
\begin{equation}\label{eq:bellman_off}
\voff(t,s)[\omega] \defeq \max_{a\in\Ac}\crl{\Re(a,s,\JT{t})  + \voff(t-1,\Tr(a,s,\JT{t}))[\omega] },    
\end{equation}
with boundary condition $\voff(0,s)[\omega]=0$ for all $s\in\St$. 
The notation $\voff(t,s)[\omega]$ is used to emphasize that, given sample-path $\omega$,  \off's value function is a deterministic function of $t$ and $s$.

We require that the DP formulation in \cref{eq:bellman_off} is well defined.
For simplicity, we enforce this with the following assumption: there are some constants $c_1,c_2\geq 0$ such that $-c_1 \leq \max_{a\in\Ac}\Re(a,s,j) \leq c_2$ for all $s\in\St,j\in\J$.
In other words, every state has a feasible action and the maximum reward is uniformly bounded and attained.
The spaces $\St,\J,\Ac$ and functions $\Tr,\Re$ are otherwise arbitrary.
We enforce this assumption for clarity of exposition, but we observe that it can be further generalized, c.f. \citep[Volume II, Appendix A]{bertsekas1995dynamic}.

On the other hand, \onl chooses actions based on \emph{policy} $\pon$ defined as follows:
\begin{definition}[Online Policy]
An online policy $\pon$ is any collection of functions $\crl{\pon(t,s,j): t\in [T],s\in\St,j\in\J}$ such that, if at time $t$ the current state is $s$ and a type $j$ arrives, then \onl chooses the action $\pon(t,s,j)\in\Ac$.
The function $\pon(t,\cdot,\cdot)$ can depend only on $\crl{\JT{T},\cdots,\Jt}$, i.e., on the randomness observed at periods $\tau\geq t$ (the history).
\end{definition}

Let us denote $\crl{S^t:t\in [T]}$ as \onl's state over time, i.e., the stochastic process $S^t\in\St$ that results from following a given policy $\pon$.
We can write \onl's accrued value, for a given policy $\pon$, as 
\begin{align*}
\von(t,S^t)[\omega] \defeq \sum_{\tau\leq t} \Re(\pon(\tau,S^\tau,\JT{\tau}),S^\tau,\JT{\tau})[\omega]
\end{align*}
For notational ease, we omit explicit indexing of $\von$ on policy $\pon$.

On any sample-path $\omega$, we can define the \emph{regret} of an online policy to be the \emph{additive} loss incurred by \onl using $\pon$ w.r.t. \off, i.e., 
$$\reg[\omega]\defeq\voff(T,S^T)[\omega]-\von(T,S^T)[\omega]$$

\begin{remark}[Regret is Agnostic of \off Algorithm]\label{rem:obliv}
Our previous definition of $\reg$ depends only on the online policy $\pon$, but it does not depend on the policy (or algorithm) used by \off as long as it is optimal.
For example, in the case where there are multiple maximizers in the Bellman \cref{eq:bellman_off}, different tie-breaking rules for \off yield different algorithms, but all of them are optimal and have the same optimal value $\voff$.
\end{remark}

\subsection{The Compensated Coupling}
\label{ssec:compcoupling}
At a high-level, the compensated-coupling is a sample path-wise charging scheme, wherein we try to couple the trajectory of a given policy to a sequence of offline policies. 
Given any non-anticipatory policy (played by \onl), the technique works by making \emph{\off follow \onl} -- formally, we couple the actions of \off to those of \onl, while compensating \off to preserve its collected value along every sample-path. 

\begin{myexample}
Consider the multi-secretary problem with budget $B=1$ and three arriving types $\J=\crl{1,2,3}$ with $r_1>r_2>r_3$. 
The state space in this problem is $\St=\N$ and the action space is $\Ac=\crl{\text{accept},\text{reject}}$.
Suppose for $\T=4$ the arrivals on a given sample-path are $(\JT{4},\JT{3},\JT{2},\JT{1})=(1,2,1,3)$. 
Note that \off  will accept exactly one arrival of type $1$, but is indifferent to which arrival. 
While analyzing \onl, we have the freedom to choose a benchmark by specifying the tie-breaking rule for \off -- for example, we can compare \onl to an \off agent who chooses to \emph{front-load} the decision by accepting the arrival at $t=4$ (i.e., as early in the sequence as possible) or \emph{back-load} it by accepting the arrival at $t=2$.
In conclusion, for this sample path, the following two sequences of actions are optimal for \off: (accept,reject,reject,reject) and (reject,reject,accept,reject).

Suppose instead that we choose to reject the first arrival ($t=4$), and then want \off to accept the type-$2$ arrival at $t=3$ -- this would lead to a decrease in \off's final reward. 
The crucial observation is that we can still \emph{incentivize} \off to accept arrival type $2$ by offering a \emph{compensation} (i.e., additional reward) of $r_1-r_2$ for doing so.
The basic idea behind the compensated coupling is to generalize this argument.
We want \off to take \onl's action, hence we couple the states of \off and \onl with the use of compensations. 
\end{myexample}

We start with a general problem:
Given sample-path $\omega\in\Omega$ with arrivals $\crl{\Jt[\omega]: t\in[\T]}$, recall $\voff(t,s )[\omega]$ denotes \off's value starting from state $s$ with $t$ periods to go.
$\voff(t,s )[\omega]$ obeys the Bellman \cref{eq:bellman_off}.
The following definition is about the actions satisfying said Bellman Equations.

\begin{definition}[Satisfying Action] \label[definition]{def:satisfying}
Fix $\omega\in\Omega$.
For any given state $s$ and time $t$, we say \off is \emph{satisfied with an action $a$} at $(s,t)$ if $a$ is a maximizer in the Bellman equation, i.e., $$a\in\argmax_{\hat{a}\in\Ac}\left\{\Re(\hat{a},s,\JT{t}) +\voff(t-1,\Tr(\hat{a},s,\JT{t}) )[\omega]\right\}.$$
\end{definition}
Observe that $a$ may be satisfying for a sample path $\omega$ and not for some other $\omega'$; once the sample path is fixed, satisfying actions are unequivocally identified.

\begin{myexample}
Consider the multi-secretary problem with  $\T=5$, initial budget $B=2$, types $\J = \{1,2,3\}$ with $r_1>r_2>r_3$, and a particular sequence of arrivals $(\JT{5},\JT{4},\JT{3},\JT{2},\JT{1})=(2,3,1,2,3)$. 
The optimal value of \off is $r_1+r_2$, and this is achieved by accepting the sole type-$1$ arrival as well as any one out of the two type-$2$ arrivals.
At time $t=5$, \off is satisfied either accepting or rejecting $\JT{5}$.
Further, at $t=3$, for any budget $b>0$ the only satisfying action is to accept.
\end{myexample}

With the notion of satisfying actions, we can create a coupling as illustrated in \cref{fig:diagram}.
Although \off may be satisfied with multiple actions (see above example and \cref{rem:obliv}), its value remains unchanged under any satisfying action, i.e., any tie-breaking rule.
We define a valid policy $\poff$ for \off to be any \emph{anticipatory functional} such that, for every $\omega\in \Omega$, we have a different mapping to actions.
Formally, for every $\omega\in\Omega$, $\poff[\omega]:[\T]\times\St\times\J \to \Ac$ is a function satisfying the optimality principle:
\begin{align*}
\voff(t,s )[\omega]=\voff(t-1,\Tr(\poff(t,s,\Jt)[\omega],s,\Jt)) [\omega] +  
\Re(\poff(t,s,\Jt)[\omega],s,\JT{t}), \quad \forall t\in [T],s\in \St.
\end{align*} 

Next, we quantify by how much we need to compensate \off when \onl's action is not satisfying, as follows
\begin{definition}[Marginal Compensation]
\label[definition]{def:margincomp}
For action $a\in\Ac$, time $t\in[\T]$ and state $s\in\St$, we denote the random variable $\Rl$ and scalar $\rl$
\begin{align*}
\Rl(t,a,s ) &\defeq \voff(t,s ) -[\voff(t-1,\Tr(a,s,\Jt) )+\Re(a,s,\Jt)]\\
\rl(a,j) &\defeq \max\crl*{\Rl(t,a,s )[\omega]:t\in[\T],s\in\St, \omega\in\Omega\mbox{ s.t. } \Jt[\omega]=j}.
\end{align*}
\end{definition}
The random variable $\Rl$ captures exactly how much we need to compensate \off, while $\rl(a,j)$ provides a uniform (over $s,t$) bound on the compensation required when \onl errs on an arrival of type $j$ by choosing an action $a$. 
Though there are several ways of bounding $\Rl(t,a,s )$, we choose $\rl(a,j)$ as it is clean, expressive, and admits good bounds in many problems as the next example shows.

\begin{myexample}
For Online Packing problems define $\rmax\defeq\max_{j\in [n]}r_j$ as the maximum reward over all classes.
The state space in this problem is $\St=\N^d$ and the actions space is $\Ac=\crl{\text{accept},\text{reject}}$.
Also, for simplicity, we assume that all resource requirements are binary, i.e., $a_{ij}\in\{0,1\}\,\forall\,i\in[d],j\in[n]$. 
For a given sample-path $\omega\in\Omega$ and any given budget $b\in\N^d$, if \off decides to accept the arrival at $t$, we can instead make it reject the arrival while still earning a greater or equal reward by paying a compensation of $\rmax$.
On the other hand, note that \off can at most extract $\rmax$ in the future for every resource $\Jt$ uses; hence on sample-paths where \off wants to reject $\Jt$, it can be made to accept $\Jt$ instead with a compensation of $\norm{A_{\Jt}}_1\rmax\leq d\rmax$.
In conclusion, we have $r_j\leq \rl(a,j)\leq d\rmax$.
\end{myexample}

Recall that $S^t$ denotes the random process of \onl's state.
Additionally, we denote $\Ron(t,S^t)[\omega]\defeq \Re(\pon(t,S^t,\theta^t),S^t,\Jt)$ as the reward collected by \onl at time $t$, and hence $\von(\T,S^\T )[\omega]=\sum_{t\in[T]}\Ron(t,S^t)[\omega]$. 

The final step is to fix \off's policy to be one which `follows \onl' as closely as possible.
For this, given a policy $\pon$, on any sample-path $\omega$ we set $\poff(t,s,\Jt)[\omega]=\pon(t,s,\Jt)[\omega]$ if $\pon(t,s,\Jt)[\omega]$ is satisfying, and otherwise set $\poff(t,s,\Jt)[\omega]$ to an arbitrary satisfying action.
In other words, we start with any valid policy $\poff$ and, for every $\omega\in \Omega$, we modify it as described to obtain another valid policy.
Abusing notation, we still call $\poff$ this modified policy.
Recall that this modification does not change the regret guarantees (see \cref{rem:obliv}).

\begin{definition}[Disagreement Set]\label{def:disagreement}
For any state $s$ and time $t$, and any action $a\in \Ac$, we define the \emph{disagreement set} $Q(t,a,s)$ to be the set of sample-paths where $a$ is not satisfying for \off, i.e., 
\begin{align*}
Q(t,a,s)\defeq\crl*{\omega\in\Omega :\voff(t,s )[\omega]>\Re(a,s,\Jt) +  
 \voff(t-1,\Tr(a,s,\Jt) )[\omega]}.
\end{align*}
\end{definition}
Finally, let $Q(t,s)\subseteq \Omega$ be the event when \off cannot follow \onl, i.e., $Q(t,s)\defeq Q(t,\pon(t,s,\Jt),s)$.
Note that $Q(t,s)$ depends on $\pon$, but we omit the indexing since $\pon$ is clear from context. 
\emph{Only under $Q(t,s)$ we need to compensate \off}, hence we obtain the following.

\begin{figure}
\includegraphics[width=0.45\textwidth]{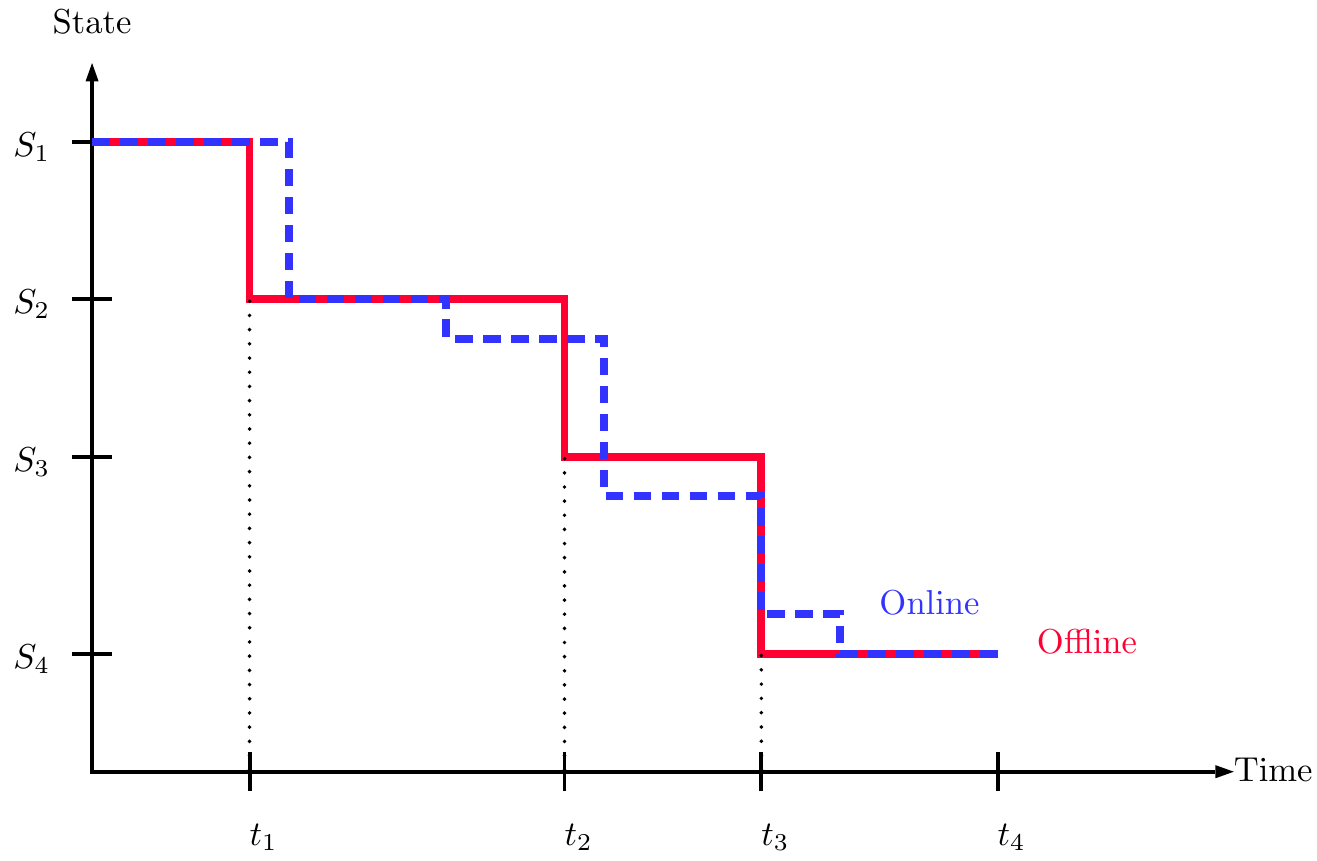}%
\includegraphics[width=0.45\textwidth]{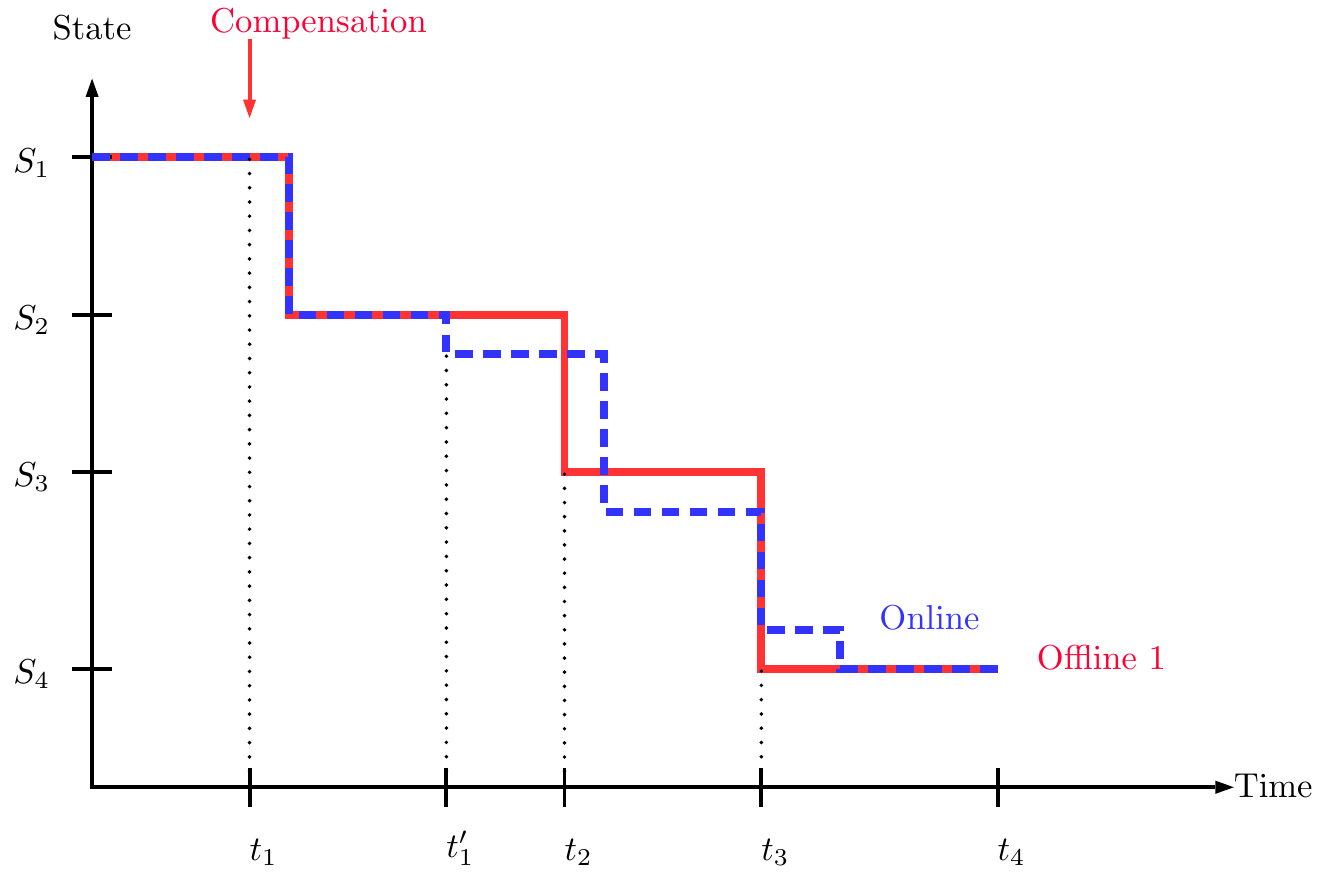}

\includegraphics[width=0.45\textwidth]{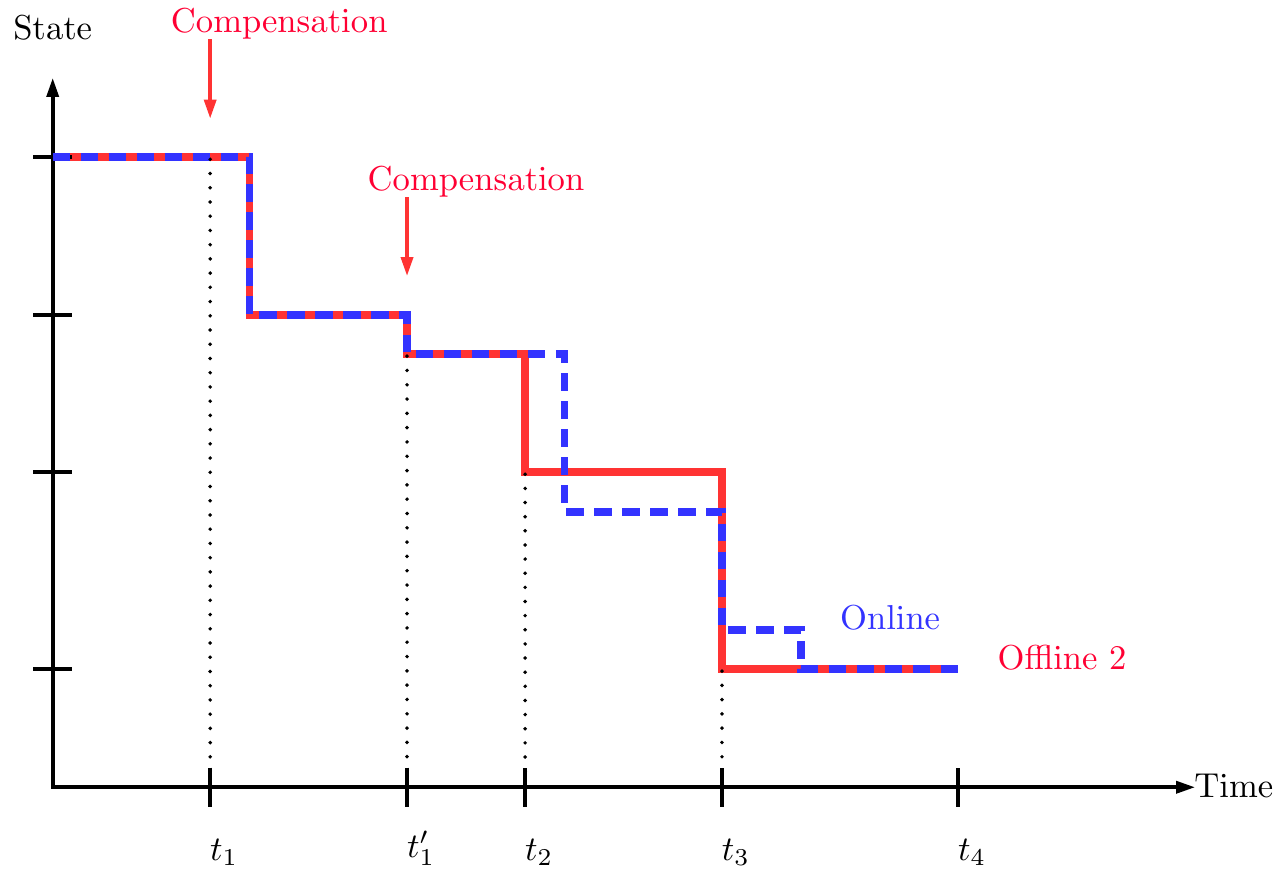}%
\includegraphics[width=0.45\textwidth]{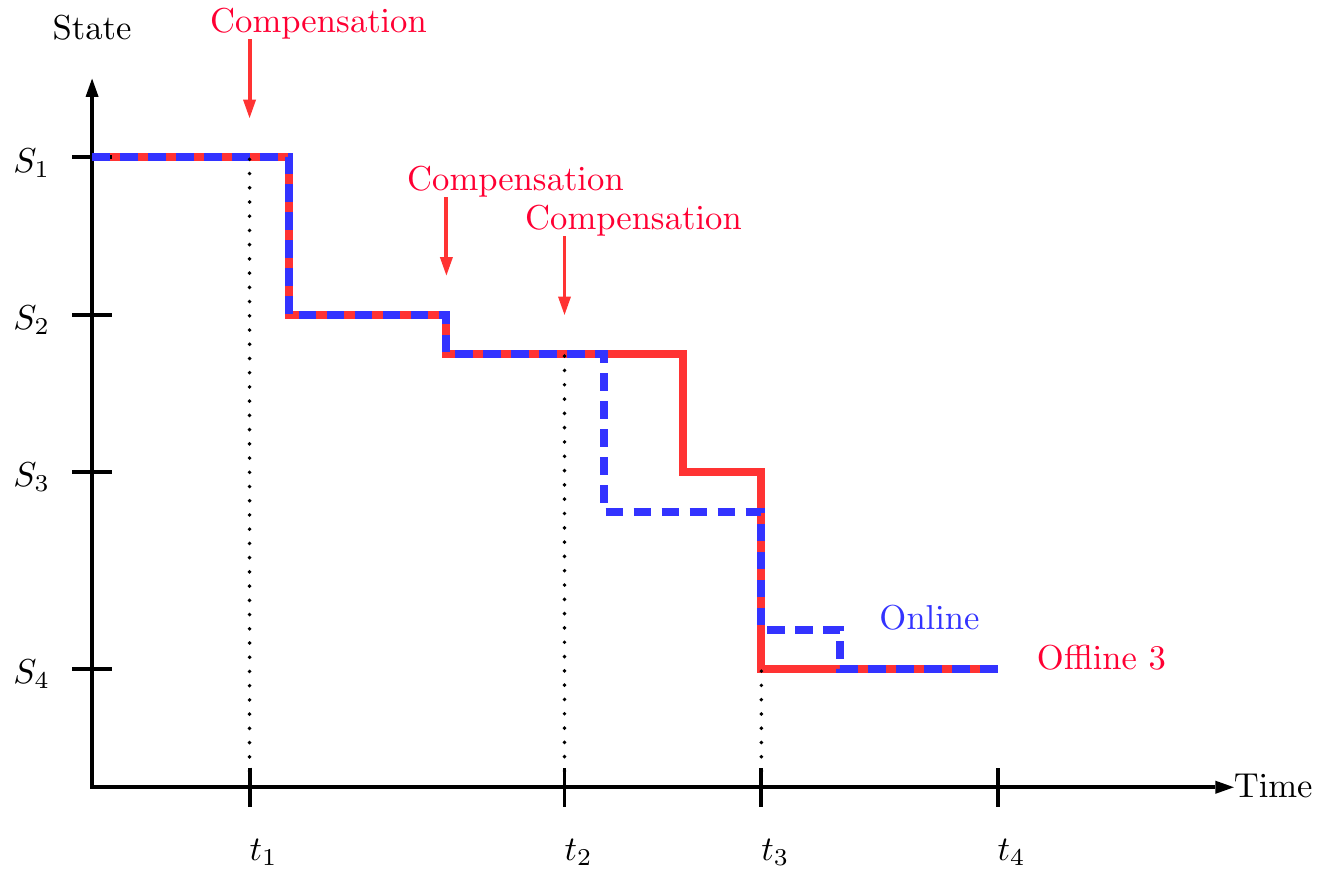}
\caption{The top-left image shows the traditional approach to regret analysis, wherein one considers a fixed offline policy (which here corresponds to a fixed trajectory characterized by accept decisions at $t_1,t_2,t_3,\ldots$) and tries to bound the loss due to ``\onl oscillating around \off''.
In contrast, the compensated coupling approach compares \onl to an \off policy which changes over time. This leads to a sequence of offline trajectories (top-right, bottom-left, and bottom-right), each ``agreeing'' more with \onl.
In particular, \off is not satisfied with \onl's action at $t_1$ (leading to divergent trajectories in the top-left figure), but is made to follow \onl by paying a compensation (top-right), resulting in a new \off trajectory, and a new disagreement at $t_1'\in(t_1,t_2)$. This coupling process is repeated at time $t_1'$ (bottom-left), and then at $t_2$ (bottom-right), each time leading to a new future trajectory for \off.
Coupling the two processes helps simplify the analysis as we now need to study a single trajectory (that of \onl), as opposed to all potential \off trajectories.
}
\label{fig:diagram}
\end{figure}

\begin{lemma}[Compensated Coupling]
\label[lemma]{lem:coupling}
For any online decision-making problem, fix any \onl policy $\pon$ with resulting state process $S^t$.
Then we have:
\begin{align*}
\reg[\omega] = \sum_{t\in[T]} \Rl(t,\pon(t,S^t,\Jt),S^t )[\omega]\cdot \Ins{Q(t,S^t)}[\omega],
\end{align*}
and thus
$\E[\reg] \leq \max_{a\in\Ac,j\in\J}\left\{\rl(a,j)\right\}\cdot\sum_{t\in[T]}\E[\Pr[Q(t,S^t)|S^t]]$.
\end{lemma}
\begin{proof}
We stress that, throughout, $S^t$ denotes \onl's state.
We claim that, for every time $t$,
\begin{align}
\voff(t,S^t )[\omega]-\voff(t-1,S^{t-1} )[\omega] 
= \Ron(t,S^t )[\omega] + \Rl(t,\pon(t,S^t,\Jt),S^t )[\omega]\Ins{Q(t,S^t)}[\omega]. \label{eq:compensation}
\end{align}
To see this, let $a =\pon(t,S^t,\Jt)$. 
If \off is satisfied taking action $a$ in state $S^t$, then $\voff(t,S^t )[\omega]-\voff(t-1,S^{t-1} )[\omega]  = \Ron(t,S^t )[\omega]$.
On the other hand, if \off  is not satisfied taking action $a$, then by the definition of marginal compensation (\cref{def:margincomp}) we have,
$\voff(t,S^t )[\omega] -\voff(t-1,\Tr(a,S^t,\Jt) )[\omega] = \Rl(t,a,S^t )[\omega]+ \Re(a,S^t,\Jt)$.
Since by definition $\Tr(a,S^t,\Jt)=S^{t-1}$ and $\Re(a,S^t,\Jt)=\Ron(t,S^t)$, we obtain~\cref{eq:compensation}.
Finally, our first result follows by telescoping the summands and the second by linearity of expectation.
\end{proof}

We list a series of remarks.
\begin{itemize}[nosep,leftmargin=0.5cm]
\item \cref{lem:coupling} is a sample-path property that makes no reference to the arrival process. Though we use it primarily for analyzing MDPs, it can also be used for adversarial settings.
We do not further explore this, but believe it is a promising avenue.
\item For stochastic arrivals, the regret depends  on $\E\brk*{\sum_{t\in[T]}\Pr[Q(t,S^t)]}$;
it follows that, if the disagreement probabilities are summable over all $t$, then the expected regret is constant. In \cref{sec:regret_rm,sec:matching_type} we show how to bound $\Pr[Q(t,S^t)]$ for different problems.
\item The first part of \cref{lem:coupling} provides a distributional characterization of the regret in terms of a weighted sum of Bernoulli variables. This allows us to get high-probability bounds in \cref{ssec:highprob}.
\item \cref{lem:coupling} gives a tractable way of bounding the regret which does not require either reasoning about the past decisions of \onl, or the complicated process \off may follow.
In particular, it suffices to bound $\Pr[Q(t,S^t)]$, i.e., the probability that, \emph{given} state $S^t$ at time $t$, \off loses optimality in trying to follow \onl.

\item As mentioned before, \cref{lem:coupling} extends to the full generality of MDPs.
Indeed, from \citep[chapter 11]{online_book}, it follows that any MDP with random transitions and random rewards can be simulated by the family of MDPs we study here; since the inputs $\Jt$ are allowed to be random, we can define random transitions and rewards based on $\Jt$, see \citep{online_book} for further details.
\end{itemize}

\cref{lem:coupling} thus gives a generic tool for obtaining regret bounds against the offline optimum for any online policy. 
Note also that the compensated coupling argument generalizes to settings where the transition and reward functions are time dependent, the policies are random, etc.
The compensated coupling also suggests a natural greedy policy, which we define next. 

\paragraph{Comparison to traditional approaches.}
As discussed in related work (\cref{sec:related}), there are two main approaches.
First, the fluid (or ex ante) benchmark, which can be understood as competing against a fixed value.
This is the prevailing technique in competitive analysis \citep{alaei_bayesian} and the Online Packing literature \citep{jasin2012,Reiman_nrm,wu2015algorithms}.
We showed in \cref{prop:bad_fluid} that such an approach cannot yield better than $O(\sqrt{T})$ regret bounds, while we prove $O(1)$.
Second, the traditional sample path approach, which competes against the random trajectory of \off (as illustrated in \cref{fig:diagram}), is based on showing that \onl is ``close'' to a fixed trajectory.
This approach is capable of obtaining strong $O(1)$ guarantees \citep{itai_secretary}, but it is highly involved since it necessitates a complete characterization of \off's trajectory.
The benefit of our approach is abstracting away from the characterization of \off's trajectory and focusing only on \onl's (which is the one that the algorithm controls), while yielding strong $O(1)$ guarantees.

The next example illustrates the Compensated Coupling in a different setting.
We consider the Ski Rental problem, which is a well studied minimum cost covering (not packing) problem.
\begin{myexample}[Ski Rental]
Given $T$ days for skiing, each day we decide whether to buy skis for $b$ dollars or to rent them for $1$ dollar.
The snow may melt any day and we have a distribution over the period we may be able to ski, i.e., there is snow during the first $X\in[T]$ periods and we know the distribution of $X$.
Our aim is to explicitly write the regret of a particular policy (stated later) using the Compensated Coupling.

The optimal offline solution is trivial: if $X<b$, \off rents every day, otherwise ($X\geq b$) \off buys the first day.
In other words, \off either buys the first day (which has cost $b$), or rents every day, with cost $X$. 
Since \off knows $X$, he picks the minimum.

We map this problem to our framework as follows.
The state space is $\St=\crl{\text{skis},\text{no-skis}}$, where `skis' means we own the skis.
Arrivals are signals $\Jt\in\crl{0,1}$, where $1$ means there is snow and $0$ that the season is over.
The arrival sequence is always of the form $(\JT{T},\ldots,\JT{1})=(1,\ldots,1,0,\ldots,0)$, where $X=\sum_{t\in [T]}\Jt$ by definition.
Finally, rewards are $-1$ per day if we rent and $-b$ when we buy. 

The compensations are as follows.
If the state is `skis' or if $\Jt=0$ (season is over), then no compensation is needed, because we know that \off does nothing w.p. 1 (either he owns the skis or the problem ended).
The only case where \onl and \off may disagree is when $\Jt=1$ (can ski today) and the state is `no-skis'.

Let us denote $\Xt$ as the number of remaining skiing days (including $t$) and say we observe $\Jt=1$ at time $t$.
\off is not satisfied renting if $\Xt>b$; forcing him to rent in this event requires a compensation of $1$.
On the other hand, \off is not satisfied buying if $\Xt<b$; forcing him needs a compensation of $b-\Xt$.
Consider the following policy $\pon$: for fixed $\tau\geq 0$, rent the first $\tau$ days and buy on day $\tau+1$ contingent on seeing snow all these days, i.e., contingent on $\JT{t}=1$ for all $t=T,\ldots,T-\tau$. 
The Compensated Coupling (\cref{lem:coupling}) allows us to write
\begin{align*}
\reg = \sum_{t\in [T]}\Rl(t,\pon(t,S^t,\Jt),S^t )[\omega]\cdot \Ins{Q(t,S^t)}[\omega]
= \sum_{t=T-\tau+1}^T\In{\Xt>b} + (b-X^{T-\tau})\In{1\leq X^{T-\tau}<b}.
\end{align*}
The first term corresponds to the disagreement of the first $\tau$ days (pay one dollar each day $t$ such that $\Xt>b$), whereas the second is the disagreement of day $\tau+1$.
This is an example where the compensated coupling yields an intuitive way of writing the regret.
Furthermore, since the expression is exact, we can take expectations and optimize over $\tau$ to get the optimal policy (for example, see~\citep{power_down}).
\end{myexample}

\subsection{The Bayes Selector Policy}
\label{sec:bayes_selector}

Using the formalism defined in the previous sections, let $q(t,a,s)\defeq\Pr[Q(t,a,s)]$ be the \emph{disagreement probability} of action $a$ at time $t$ in state $s$ (i.e., the probability that $a$ is not a satisfying action).

For any $t,s,\Jt$, suppose we have an \emph{oracle} that gives us $q(t,a,s)$ for every feasible action $a$.
Given oracle access to $q(t,a,s)$ (or more generally, over-estimates $\hat q$ of $q$), a natural greedy policy suggested by \cref{lem:coupling} is that of choosing action $a$ that minimizes the probability of disagreement. 
This is similar in spirit to the \emph{Bayes selector} (i.e., hard thresholding) in statistical learning. 
\cref{alg:bayes_selector} formalizes the use of this idea in online decision-making.
The results below are essentially agnostic of how we obtain this oracle.

\begin{algorithm}
\caption{Bayes Selector}
\label{alg:bayes_selector}
\begin{algorithmic}[1]
\Require Access to over-estimates $\hat q(t,a,s)$ of the disagreement probabilities, i.e. $\hat q(t,a,s)\geq q(t,a,s)$
\Ensure Sequence of decisions for \onl.
\State Set $S^\T$ as the given initial state
\For{$t=\T,\ldots,1$}
	\State Observe the arriving type $\Jt$ .
	\State Take an action minimizing disagreement, i.e., $A^t\in\argmin\crl{\hat q(t,a,S^t): a\in\Ac}$.
	\State Update state $S^{t-1}\gets \Tr(A^t,S^t,\Jt)$.
\EndFor
\end{algorithmic}
\end{algorithm}
From \cref{lem:coupling}, we immediately have the following:

\begin{corollary}[Regret Of Bayes Selector]\label[corollary]{cor:bayes}
Consider \cref{alg:bayes_selector} with over-estimates $\hat q(t,a,s)\geq \Pr[Q(t,a,s)]\,\forall\,(t,a,s)$.
If $A^t$ denotes the policy's action at time $t$, then
\begin{align*}
\E[\reg]
\leq \max_{a\in\Ac,j\in\J}\rl(a,j) \cdot\sum_{t\in[T]}\E[\hat q(t,A^t,S^t)].
\end{align*}
\end{corollary}

The next result states that, if we can bound the estimation error uniformly over states and actions, then the guarantee of the algorithm increases additively on the error (not multiplicatively, as one may suspect).
In more detail, our next result is agnostic of the oracle used to obtain the the estimators $\hat q$.
Examples of estimation procedures to obtain $\hat q$ include: simulation, function approximation, neural networks, etc.
Regardless of how $\hat q$ is obtained, we can give a regret guarantee based only on the accuracy of the estimators.
The following result follows as a special case of \cref{cor:bayes}; we state it to emphasize that $\hat q$ can be estimated with some error.

\begin{corollary}[Bayes Selector with Imperfect Estimators]\label[corollary]{cor:additive}
Assume we have estimators $\hat q(t,a,s)$ of the probabilities $q(t,a,s)$ such that $\abs{q(t,a,s)-\hat q(t,a,s)}\leq \Delta^t$ for all $t,a,s$.
If we run \cref{alg:bayes_selector} 
with over-estimates $\hat q(t,a,s) + \Delta^t$, and $A^t$ denotes the policy's action at time $t$, then
\begin{align*}
\E[\reg]
\leq \max_{a\in\Ac,j\in\J}\rl(a,j) \cdot\sum_{t\in[T]}(\E[\hat q(t,A^t,S^t)]
+ \Delta^t).
\end{align*}
\end{corollary}

Observe that, the total error induced due to estimation is a constant if, e.g., we can guarantee $\Delta^t=1/t^2$ or $\Delta^t=1/(T-t)^2$.

It is natural to consider a more sophisticated version of \cref{alg:bayes_selector}, wherein we make decisions not only based on disagreement probabilities, but also take into account marginal compensations, i.e., the marginal loss of each decision.
While \cref{alg:bayes_selector} is enough to obtain constant regret bounds in the problems we consider, we note that such an extension is possible and can be found in \cref{sec:loss}.

\begin{remark}
We discussed in \cref{sec:coupling} that the Compensated Coupling extends to the case where transitions and rewards can be random.
By the same argument, the Bayes Selector (\cref{alg:bayes_selector}) and its guarantees (\cref{cor:bayes,cor:additive}) extend too.
Notice that \off here is a prophet who has full knowledge of all the randomness (arrivals, transitions, and rewards).
\end{remark}

\section{Regret Guarantees for Online Packing}
\label{sec:regret_rm}

We now show that for the Online Packing problem, the Bayes Selector achieves an expected regret which is \emph{independent of the number of arrivals $T$ and the initial budgets $B$}; in~\cref{sec:matching_type}, we extend this to Online Matching problems.

In more detail, we prove that the \emph{dynamic fluid relaxation} $(P_t)$ in \cref{eq:coupled_lp} provides a good estimator for the disagreement probabilities $q(t,a,s)$, and moreover, that the Bayes Selector based on these statistics reduces to a simple \emph{re-solve and threshold} policy.

In this setting, the state space corresponds to resource availability, hence $\St=\N^d$; there are two possible actions, accept or reject, hence $\abs{\Ac} = 2$; finally, transitions correspond to the natural budget reductions given by the matrix $A$.

Recall that $Z(t)\in\N^n$ denotes the cumulative arrivals in the last $t$ periods and $B^t\in\N^d$ denotes \onl's budget at time $t$.
Given knowledge of $Z(t)$ and state $B^t$, we define the \emph{ex-post} relaxation $(P^\star_t)$ and \emph{fluid} relaxation $(P_t)$ as follows.
\begin{equation}\label{eq:coupled_lp}
\begin{array}{rrl}
(P_t^\star)\, \max  & r'x& \\
\text{s.t.}&  Ax&\leq B^t \\
& x&\leq Z(t) \\
& x&\geq 0.
\end{array}
\qquad\qquad
\begin{array}{rrl}
(P_t)\, \max  & r'x& \\
\text{s.t.}&  Ax&\leq B^t \\
& x&\leq \E[Z(t)] \\
& x&\geq 0.
\end{array}
\end{equation}

\begin{remark}
\off solves $(P^\star_t)$ in \cref{eq:coupled_lp}, while \onl solves $(P_t)$.
Both problems depend on \onl's budget at $t$; this is a crucial technical point and can only be accomplished due to the coupling we have developed.
\end{remark}

Let $\Xt$ be a solution of $(P_t)$ and $\Xts$ a solution of $(P_t^\star)$.
Uniqueness of solutions is not required, see \cref{prop:lipschitz}, and $\Xts$ is for the analysis only.
Our policy is detailed in \cref{alg:fluid}.

\begin{algorithm}
\caption{Fluid Bayes Selector}
\label{alg:fluid}
\begin{algorithmic}[1]
\Require Access to solutions $\Xt$ of $(P_t)$ and resource matrix $A$.
\Ensure Sequence of decisions for \onl.
\State Set $B^\T$ as the given initial budget levels
\For{$t=\T,\ldots,1$}
	\State Observe arrival $\Jt=j$ and accept iff $\Xt_j\geq \E[Z_j(t)]/2$ and it is feasible, i.e., $A_j\leq B^t$.
	\State Update $B^{t-1}\gets B^t-A_j$ if accept, and $B^{t-1}\gets B^t$ if reject.
\EndFor
\end{algorithmic}
\end{algorithm}

Intuitively, we `front-load' (accept as early as possible) classes $j$ such that $\Xt_j\geq \E[Z_j(t)]/2$ and back-load the rest (delay as much as possible).
If \off is satisfied accepting a front-loaded class (resp. rejecting a back-loaded class), he will do so.
Accepting class $j$ is therefore an error if \off, given the same budget as \onl, picks no future arrivals of that class (i.e., $\Xts_j<1$).
On the other hand, rejecting $j$ is an error if $\Xts_j>Z_j(t)-1$.
We summarize this as follows:
\begin{enumerate}
	\item Incorrect rejection: if $\Xt_j< \frac{\E[Z_j(t)]}{2}$ and $\Xts_j>Z_j(t)-1$.
	\item Incorrect acceptance: if $\Xt_j\geq \frac{\E[Z_j(t)]}{2}$ and $\Xts_j<1$.
\end{enumerate}
Observe that a compensation is paid only when the fluid solution is far off from the correct stochastic solution.
In the proofs of \cref{theo:secretary,theo:reg_general} we formalize the fact that, since $\Xt$ estimates $\Xts$, such an event is highly unlikely -- this along with the compensated coupling provides our desired regret guarantees.

\paragraph{Disagreement probabilities and the Fluid Bayes Selector.}
The Bayes Selector (\cref{alg:bayes_selector}) runs with over-estimates $\hat q(t,a,b)$ and, at each time, picks the minimum.
On the other hand, \cref{alg:fluid} is presented as the ``simplified version'', in the sense that the decision rule is the one that minimizes suitable over-estimates $\hat q$.
More importantly, we prove the following properties
\begin{align}
\argmin\crl{\hat q_j(t,\text{accept},b),\hat q_j(t,\text{reject},b)} &= \left\{
\begin{array}{ll}
\text{accept}     & \text{ if } \Xt_j\geq \E[Z_j(t)]/2  \\
\text{reject}     & \text{ if } \Xt_j<\E[Z_j(t)]/2  
\end{array}
\right., \label{eq:fluid_rule}\\
\min\crl{\hat q_j(t,\text{accept},b),\hat q_j(t,\text{reject},b)} &\leq c_1 e^{-c_2t}. \label{eq:bound_q}
\end{align}
In other words, the property in \cref{eq:fluid_rule} shows that \cref{alg:fluid} is a Bayes Selector, while the property in \cref{eq:bound_q} yields the desired constant regret bound in virtue of the Compensated Coupling and \cref{cor:bayes}. 
We give explicit expressions for the values $\hat q$ and constants $c_1,c_2$, see e.g.\ \cref{eq:q_secre}.

\paragraph{The robustness of the Bayes Selector.}
The probability minimizing disagreement can be uniformly bounded over all budgets $b\in\N^d$, i.e., the exponential bound in \cref{eq:bound_q} does not depend on $b$.
This property has the following consequence: \emph{since the Fluid Bayes Selector has strong performance, many other Bayes Selector algorithms (using different $\hat q$) do too.}
In other words, the design of algorithms based on the Bayes Selector is robust and does not depend on ``fine tunning'' of the parameters $\hat q$.
We make this precise in \cref{cor:secretary,cor:general} and uncover the same phenomenon for matching problems, see \cref{cor:matching}.

We need some additional notation before presenting our results.
Let $\E_j[\cdot]$ ($\Pr_j[\cdot]$) be the expectation (probability) conditioned on the arrival at time $t$ being of type $j$, i.e., \ $\Pr_j[\cdot] = \Pr[\cdot|\Jt=j]$.
We denote $\rmax\defeq \max_{j\in [n]}r_j$ and $\pmin\defeq\min_{j\in [n]}p_j$.

\subsection{Special Case: Multi-Secretary with Multinomial Arrivals}
\label{sec:warm_up}

Before we proceed to the general case, we state the result for the Multi-Secretary problem.
We present this result separately because in this one dimensional problem we can obtain a better and explicit constant.
The proofs of \cref{theo:secretary} and \cref{cor:secretary} below can be found in Appendix \ref{sec:app_multi}.

\begin{theorem}\label{theo:secretary}
The expected regret of the Fluid Bayes Selector (\cref{alg:fluid}) for the multi-secretary problem with multinomial arrivals is at most $\rmax\sum_{j>1}2/p_j\leq 2(n-1)\rmax/\pmin$.     
\end{theorem}

This recovers the best-known expected regret bound for this problem shown in a recent work~\citep{itai_secretary}.
However, while the result in \citep{itai_secretary} depends on a complex martingale argument, our proof is much more succinct, and provides explicit and stronger guarantees; in particular, in Section~\ref{ssec:highprob}, we provide concentration bounds for the regret.

Moreover,  \cref{theo:secretary}, along with \cref{cor:additive}, provides a critical intermediate step for characterizing the performance of \cref{alg:bayes_selector} for the multi-secretary problem.

\begin{corollary}\label[corollary]{cor:secretary}
For the multi-secretary problem with multinomial arrivals, the expected regret of the Bayes Selector (\cref{alg:bayes_selector}) with any imperfect estimators $\hat q$ is at most $2\rmax\prn*{\sum_{j>1}1/p_j+\sum_{t\in[T]}\Delta^t}$, where $\Delta^t$ is the accuracy defined by $\abs{q(t,a,b)-\hat q(t,a,b)}\leq \Delta^t$ for all $t\in [T],a\in\Ac,b\in \N$.
\end{corollary}

Observe that, if $\Delta^t$ is summable, e.g., $\Delta^t=1/t^2$ or $\Delta^t=1/(T-t)^2$, then \cref{cor:secretary} implies constant expected regret for all these types of estimators we can use in \cref{alg:bayes_selector}.

\subsection{Online Packing with General Arrivals}
\label{sec:warm_up2}

We consider now the case $d>1$ and arrival processes other than Multinomial.
We assume the following condition on the process $Z(t)$, which we refer to as \emph{All-Time Deviation}.

\begin{definition}[All-Time Deviation]\label[definition]{def:all_time}
Let $\mu$ be a given norm in $\R^n$ and $\kappa\in\Rp^n$ a constant parameter.
An $n$ dimensional process $Z(t)$ satisfies the All-Time Deviation bound w.r.t.\ $\mu$ and $\kappa$ if, for all $j\in [n]$, there are constants $c_j\geq 0$ and naturals $\tau_j$ such that
\begin{equation}
\label{eq:tail_z}
\Pr\brk*{\norm{Z(t)-\E(Z(t))}_\mu \geq \frac{\E[Z_j(t)]}{2\kappa_j}} \leq \frac{c_j}{t^2} \quad \forall t >\tau_j.
\end{equation}
\end{definition}

We remark that we do not need exponential tails, as it is common to assume, but rather a simple quadratic tail.
Additionally, some common tail bounds are valid only for large enough samples; the parameters $\tau_j$ capture this technical aspect. 
In this section we will use the definition with $\kappa_j$ the same entry for all $j$, thus denoted simply by $\kappa>0$.
In \cref{sec:matching_type} we require the definition with the more general form.

\begin{myexample}[Multinomial and Poisson tails]
In these examples we actually have the stronger exponential tails, so we do not elaborate on the constant $c_j$.

For multinomial arrivals,  \citep[Lemma~3]{multinomial_bound} guarantees 
\begin{equation}
\label{eq:tail_multi}
\Pr[\norm{Z(t)-\E(Z(t))}_1>t\varepsilon]\leq e^{-t\varepsilon^2/25}, \, \forall 0<\varepsilon<1, t\geq \frac{\varepsilon^2n}{20}.
\end{equation}
By setting $\varepsilon=p_j/2\kappa$, we conclude that \cref{def:all_time} is satisfied  with constants $\tau_j=(p_j/2\kappa)^2n/20$.

For Poisson arrivals, from the proof of \citep[Lemma~3]{multinomial_bound}, $\Pr(\abs{X-\lambda}\geq\varepsilon\lambda)\leq 2e^{-\lambda\varepsilon^2/4}$ is valid for $X\sim\Poiss(\lambda)$ and any $\varepsilon>0$.
Using this, we can simply take $\tau_j=0$.
\end{myexample}

In the remainder of this subsection we generalize our ideas to prove the following.

\begin{theorem}\label{theo:reg_general}
Assume the arrival process $(Z(t):t\in [T])$ satisfies the conditions in \cref{eq:tail_z}.
The expected regret of the Fluid Bayes Selector (\cref{alg:fluid}) for Online Packing is at most $d\rmax M$, where $M$ is independent of $T$ and $B$.
Specifically, for $\kappa=\kappa(A)$ we have
\begin{enumerate}
    \item For Multinomial arrivals: $M\leq 103\kappa^2\sum_{j\in[n]}1/p_j$.
    \item For general distributions satisfying \cref{eq:tail_z}: $M\leq\sum_{j\in [n]}p_j(2c_j+\max\crl{\tau_j,\tilde\tau_j})$, where $p_j$ is an upper bound on $\Pr[\Jt=j]$ and $\tilde\tau_j$ is such that $\E[Z_j(\tilde\tau_j)]\geq 2$, i.e., it is large enough.
\end{enumerate}
\end{theorem}

The constant $\kappa(A)$ is given by \cref{prop:lipschitz} below.
Just as before, \cref{theo:reg_general}, along with \cref{cor:additive}, provides a performance guarantee for  \cref{alg:bayes_selector}.
We state the corollary without proof, since it is identical to that of \cref{cor:secretary}.
\begin{corollary}\label[corollary]{cor:general}
For the Online Packing problem, if the arrival process satisfies the conditions in \cref{eq:tail_z}, the expected regret of the Bayes Selector (\cref{alg:bayes_selector}) with any imperfect estimators $\hat q$ is at most $d\rmax (M+2\sum_{t\in[T]}\Delta^t)$, where $M$ is as in \cref{theo:reg_general} and $\Delta^t$ is the accuracy defined by $\abs{q(t,a,b)-\hat q(t,a,b)}\leq \Delta^t$ for all $t\in [T],a\in\Ac,b\in \N^d$.
\end{corollary}

To prove~\cref{theo:reg_general}, we need to quantify how the change in the right-hand side of an LP impacts optimal solutions.
Indeed, as stated in \cref{eq:coupled_lp}, the solutions $\Xt$ and $\Xts$ correspond to perturbed right-hand sides ($\E[Z(t)]$ and $Z(t)$ respectively).
The following proposition implies that small changes in the arrivals vector do not change the solution by much and it is based on a more general result from~\citep[Theorem 2.4]{mangasarian}.

\begin{proposition}[LP Lipschitz Property]
\label[proposition]{prop:lipschitz}
Given $b\in \R^d$, and any norm $\norm{\cdot}_\mu$ in $\R^n$, consider the following LP
\begin{align*}
P(y) \quad \max\crl{r'x: Ax\leq b, 0\leq x\leq y, y\in \Rp^n}.
\end{align*}
Then $\exists$ constant $\kappa=\kappa_{\mu}(A)$ such that, for any $y,\hat y\in\Rp^n$  and any solution $x$ to $P(y)$, there exists a solution $\hat x$ solving $P(\hat y)$ such that $\norm{x- \hat x}_\infty\leq \kappa\norm{y - \hat y}_{\mu}$.
\end{proposition}

\begin{proofof}{\cref{theo:reg_general}}
Recall the two conditions derived from our decision rule:
(1) Incorrect rejection of $j$ means $\Xt_j< \E[Z_j(t)]/2$ and $\Xts_j>Z_j(t)-1$.
(2) Incorrect acceptance of $j$ means $\Xt_j\geq \E[Z_j(t)]/2$ and $\Xts_j<1$.
We have to additionally account for feasibility, i.e., we can only accept a request $j$ if $B_i^t\geq a_{ij}$ for all $i\in [d]$.
In case there are not enough resources, our decision rule is feasible if either $\Xt_j<\E[Z_j(t)]/2$ (reject) or $\Xt_j\geq 1$ (since $\Xt$ is feasible for $(P_t)$).
Only in the case $\Xt_j\geq \E[Z_j(t)]/2$ and $\Xt_j< 1$ we need to disregard our decision rule and are forced to reject; under such a condition we must pay a compensation of $r_j$.
Observe that this condition is never met if $\E[Z_j(t)]\geq 2$, i.e., it is vacuous for $t\geq \tilde\tau_j$.

The disagreement sets (\cref{def:disagreement}) are thus $Q(t,b)=\crl{\omega\in\Omega: \text{either } (1), (2) \text{ or } t<\tilde \tau_j}$, where (1) and (2) are the previous conditions.
Now we can upper bound the probability of paying a compensation as follows.
Call $E_j$ the event $\crl{\omega\in\Omega:\norm{Z(t)-\E(Z(t))}_1\leq \E[Z_j(t)]/2\kappa}$.
In this event, \cref{prop:lipschitz} implies $\abs{\Xt_j-\Xts_j}\leq \E[Z_j(t)]/2$, hence conditions (1) and (2) do not happen when $E_j$ occurs, i.e., $\Pr_j[Q(t,b)|E_j] \leq \In{t<\tilde\tau_j}$.
Observe that $\Pr[\bar E_j]\leq f_j(t)+\In{t<\tau_j}$, where $f_j(t)=c_j/t^2$ for general processes satisfying \cref{eq:tail_z} and $f_j(t)= e^{-t(p_j/2\kappa)^2/25}$ for the Multinomial process (see \cref{eq:tail_multi}).
Finally,  
\begin{align}
\label{eq:expec}
q_j(t,B^t) \leq \Pr[\bar E_j] + \Pr_j[Q(t,B^t)|E_j]
\leq \Pr[\bar E_j] +\In{t<\tilde\tau_j}
\leq f_j(t) + \In{t < \tau_j \text{ or } t<\tilde\tau_j}.
\end{align}
Summing up over time, we get
\[
\sum_{t\in[T]} q(t,B^t) \leq \sum_{j\in [n]}p_j\prn*{\sum_{t\in[T]} f_j(t) + \max\crl{\tau_j,\tilde \tau_j}}
\]
Since $\sum_{t\in[T]} 1/t^2\leq  \pi^2/6 \leq 2$, this finishes the proof for general processes.
For the case of Multinomial arrivals, we can be more refined.
Indeed, $\tilde \tau_j$ is defined by $\E[Z_j(\tilde\tau_j)]\geq 2$, i.e., $\tilde\tau_j \geq 2/p_j$ and $\tau_j = (p_j/2\kappa)^2n/20$ (see \cref{eq:tail_multi}).
From the previous equation, with the stronger exponential bound $f_j(t)= e^{-t(p_j/2\kappa)^2/25}$ we get
\begin{align*}
\sum_{t\in[T]} q(t,B^t) &\leq \sum_{j\in [n]}p_j\prn*{\frac{25}{(p_j/2\kappa)^2} + \max\crl{(p_j/2\kappa)^2n/20,2/p_j}} \\
&\leq 100\kappa^2\sum_{j\in [n]}\frac{1}{p_j}+3n.
\end{align*}
Since $n\leq \sum_{j\in [n]}\frac{1}{p_j}$, we arrive a the desired bound.
The result follows via the compensated coupling (\cref{lem:coupling}) and \cref{cor:bayes}.
\end{proofof}
\begin{remark}
In the multi-secretary problem it is easy to conclude $\kappa(A)=1$, thus this analysis recovers the same bound up to absolute constants (namely $103$ vs $2$).
The larger constant comes exclusively from the larger constants in the tail bounds of Multinomial compared to Binomial r.v.
\end{remark}

\begin{remark}
More refined bounds on $M$ can be obtained by not bounding $\Pr[\Jt=j]\leq p_j$, but rather by $\Pr[\Jt=j] \leq p_j(t)$.
For example, a time-varying version of a Multinomial process easily fits in our framework and the proof does not change.
\end{remark}

\begin{remark}
The theorem holds even under Markovian correlations (see \cref{ex:markov} below), where the distribution of $Z(t-1)$ depends on $\Jt$.
It is interesting that in this case it is impossible to run the optimal policy for even moderate instance sizes, since the state space is huge, while the Bayes Selector still offers bounded expected regret.
\end{remark}

We now give two examples of other arrival processes that satisfy the All-Time Deviation (\cref{def:all_time}).
The proofs of the bounds are short, but we relegate them to Appendix \ref{sec:app_arrival}.
We emphasize that \cref{ex:heavy} below has quadratic tails (instead of exponential), hence we term it heavy tailed.

\begin{myexample}[Markovian Arrival Processes]\label{ex:markov}
We consider the case where $\Jt$ is drawn from an ergodic Markov chain.
Let $P\in \Rp^{n\times n}$ be the corresponding matrix of transition probabilities.
The process unfolds as follows: at time $t=T$ an arrival $\JT{T}\in [n]$ is drawn according to an arbitrary distribution, then for $t=T,\ldots,2$ we have $\Pr[\JT{t-1}=j|\Jt] = P_{\Jt j}$.
Let $\nu\in\Rp^n$ be the stationary distribution.
\emph{We do not require long-run or other usual stationary assumptions; the process is still over a finite horizon $T$}.
This process satisfies All-Time Deviation with exponential tails.
Specifically, with the norm $\mu = \norm{\cdot}_\infty$, for some constants $c_j,c'$ that depend on $P$ only we have
\begin{equation}\label{eq:markov_all_time}
\Pr[\norm{Z(t)-\nu t}_\infty \geq \nu_jt/2\kappa_j] \leq n c' e^{-c_jt}, \quad \forall t\in [T], j\in [n].    
\end{equation}
\end{myexample}

\begin{myexample}[Heavy Tailed Poisson Arrivals]\label{ex:heavy}
We consider the case where the arrival process is governed by independent time varying Poisson processes with arrival rates $\lambda_j(t)>0$, which we assume for simplicity have finitely many discontinuity points (so that all the expectations are well defined).
Under the following conditions, the process satisfies the All-Time Deviation  with \emph{quadratic tails} and norm $\mu=\norm{\cdot}_\infty$.
\begin{align}
\max_{j,k\in [n]}\max_{s\in [0,t]}\frac{\lambda_j(s)}{\lambda_k(s)} &\leq g(t) \quad \forall t\geq 0 \label{eq:heavy_ratio}    \\
\min_{j\in[n]}\min_{s\in [0,t]} \lambda_j(s) &\geq g(t)f(t)\frac{\log(t)}{t}, \quad \text{ where } \quad \lim_{t\to\infty}f(t) = \infty. \label{eq:heavy_min}
\end{align}
In other words, we require $f(t) = \omega(1)$ and $g(t)$ is any function.
Intuitively, \cref{eq:heavy_ratio} guarantees that no type $j$ ``overwhelms'' all other types; observe that, when the rates are constant, this is trivially satisfied with $g(t)$ constant.
On the other hand, \cref{eq:heavy_min} controls the minimum arrival rate, which can be as small as $\omega(\log(t)/t)$.
Observe that our conditions allow for the intensity to increase closer to the end ($t=0$), i.e., we incorporate the case where agents are more likely to arrive closer to the deadline.
\end{myexample}
\subsection{High-Probability Regret Bounds}
\label{ssec:highprob}

We have proved that $\E[\reg]$ is constant for packing problems.
One may worry that this is not enough because, since it is a random variable, \reg may still realize to a large value.
We present a bound for the distribution of \reg showing that it has light tails.

\begin{proposition}\label[proposition]{prop:distr_reg}
For packing problems, there are constants $\tau$ and $c_j$ for $j\in[n]$, depending on $A,p$ and the distribution of $Z$ only, such that
\begin{enumerate}
    \item For Multinomial or Poisson arrivals: $\forall x>\tau$,  $\Pr[\reg > x]\leq \sum_{j}p_je^{-c_jx/\rmax}/c_j$.
    \item For general distributions satisfying \cref{eq:tail_z}: $\forall x>\tau$,    $\Pr[\reg > x]\leq \frac{\rmax}{x}\sum_{j}p_jc_j$.
\end{enumerate}

\end{proposition}

The proof is based on the following simple lemma.
The idea is to first bound the disagreements of our algorithm, as defined in \cref{sec:bayes_selector}.
The total number of disagreements is a sum of dependent Bernoulli variables, which we bound next.

\begin{lemma}\label[lemma]{lem:bound_bernoulli}
Let $\crl{\Xt:t\in[\T]}$ be a sequence of dependent r.v.\ such that $\Xt\sim\Ber(p_t)$ and let $\crl{q_t:t\in[\T]}$ be numbers such that $q_t\geq p_t$.
If we define $D\defeq\sum_{t=1}^\T\Xt$, then
\[
\Pr[D\geq d] \leq \sum_{t=d}^\T q_t
\]
\end{lemma}
\begin{proof}
Fix $d\in [\T]$ and observe that 
\[
\crl{\omega\in\Omega:D\geq d} \subseteq
\crl{\omega\in\Omega: \exists t\geq d, \Xt=1}.
\]
Indeed, if the condition $(\exists t\geq d, \Xt=1)$ fails, then at most $d-1$ variables $\Xt$ can be one.

Finally, a union bound shows $\Pr[D\geq d] \leq \sum_{t\geq d}\Pr[\Xt=1]$.
Since $q_t\geq p_t$, the proof is complete.
\end{proof}

\begin{proofof}{\cref{prop:distr_reg}}
As described in the previous subsections, we can write $\reg \leq \rmax D$, with $D$ the number of disagreements.
Additionally, $D$ is a sum of $\T$ Bernoulli r.v.\ $\Xt$, each with parameter bounded by $q_t$.

In the case of Multinomial and Poisson r.v., as described in \cref{sec:warm_up2}, we have exponential bounds $q_t\leq \sum_{j\in[n]}p_je^{-c_jt}$ for $t\geq\tau=\max_{j\in[n]}\tau_j$.
We conclude invoking \cref{lem:bound_bernoulli} and upper bounding $\sum_{t=x+1}^\T e^{-c_jt} \leq e^{-c_jx}/c_j$.

For general distributions, as described in \cref{sec:regret_rm}, we have the bounds $q_t\leq \sum_{j\in [n] }p_j\frac{c_j}{t^2}$ for $t\geq \tau$.
Using \cref{lem:bound_bernoulli} and bounding $\sum_{t=x+1}^\T t^{-2}\leq 1/x$ finishes the proof.
\end{proofof}

\section{Regret Guarantees for Online Matching}
\label{sec:matching_type}

We turn to an alternate setting, where each incoming arrival corresponds to a \emph{unit-demand} buyer -- in other words, each arrival wants a unit of a single resource, but has different valuations for different resources. 
This is essentially equivalent to the online bipartite matching problem with edge weights (weights correspond to rewards) where there can be multiple copies of each node.

As before, we are given a matrix $A\in\crl{0,1}^{d\times n}$ characterizing the demand for resources, which can be interpreted as the adjacency matrix in the Online Matching problem.
Define $S_j\defeq\crl{i\in[d]:a_{ij}=1}$.
If we allocate any resource $i\in S_j$ to an agent type $j$, we obtain a reward of $r_{ij}$, whereas allocating $i\not\in S_j$ has no reward.
We can allocate at most one item to each agent.

Given resource availability $B\in\N^d$ and total arrivals $Z\in\N^n$, we can formulate \off's problem as follows, where the variable $x_{ij}$ denotes the number of items $i$ allocated to agents of type $j$.
\begin{equation}\label{eq:off_matching_problem}
\begin{array}{rrll}
(P[Z,B])\quad \max  & \sum_{i,j}x_{ij}r_{ij}a_{ij}& \\
\text{s.t.}&  \sum_{j}x_{ij}&\leq B_i &\forall i\in [d] \\
& \sum_{i\in [d]}x_{ij}&\leq Z_j &\forall j\in [n] \\
& \vx&\geq 0.
\end{array}
\end{equation}

We assume that the  process $Z(t)$ satisfies the All-Time Deviation bound (see \cref{def:all_time}) w.r.t.\ the one-norm and parameters $\kappa_j=(\abs{S_j}+1)/2$.
This condition can be restated as follows.
For every $j\in [n]$, there are constants $c_j\geq 0$ and naturals $\tau_j$ such that
\begin{equation}\label{eq:tail_z2}
\Pr\brk*{\norm{Z(t)-\E(Z(t))}_1 \geq \frac{\E[Z_j(t)]}{\abs{S_j}+1}} \leq \frac{c_j}{t^2} \quad \forall t >\tau_j.
\end{equation}

We now state the main result of this section, which is based on an instantiation of the Bayes Selector.
As before, the theorem readily implies performance guarantees for \cref{alg:bayes_selector}, which we state without proof, since it is identical to that of \cref{cor:secretary}.

\begin{theorem}\label{theo:matching}
For the Online Matching problem, if the arrival process satisfies the conditions in \cref{eq:tail_z2}, then the expected regret of the Fluid Bayes Selector (\cref{alg:fluid_matching}) is at most $\rmax\sum_{j\in [n]}p_j(c_j+\tau_j)$, where $p_j$ is an upper bound on $\Pr[\Jt=j]$.
\end{theorem}

\begin{corollary}\label[corollary]{cor:matching}
For the Online Matching problem, if the arrival process satisfies the conditions in \cref{eq:tail_z2}, then the expected regret of the Bayes Selector (\cref{alg:bayes_selector})
with any imperfect estimators $\hat q$ is at most $\rmax (M+2\sum_{t\in[T]}\Delta^t)$.
The constant $M=\sum_{j\in [n]}p_j(c_j+\tau_j)$ is as in \cref{theo:matching} and $\Delta^t$ is the accuracy defined by $\abs{q(t,a,s)-\hat q(t,a,s)}\leq \Delta^t$.
\end{corollary}

\subsection{Algorithm and Analysis}

We start from the LP in \cref{eq:off_matching_problem}, then add a fictitious item $d+1$ which no agent wants with initial budget $B^\T_{d+1}=\T$; now all agents are matched, but, if we match an agent to $d+1$, there is no reward.
Using the Compensated Coupling, we can write two coupled optimization problems, $(P_t^\star)$ for \off and $(P_t)$ for \onl as follows.
\begin{equation}\label{eq:coupled_lp_matching}
\begin{array}{rrll}
(P_t^\star)\max&  \multicolumn{2}{l}{\sum_{i\in [d],j\in [n]}x_{ij}r_{ij}a_{ij} }\\
\text{s.t.}&  \sum_{j\in [n]}x_{ij}&\leq B_i^t &\forall i\in [d+1] \\
& \sum_{i\in [d+1]}x_{ij}&=Z_j(t) &\forall j\in [n] \\
& \vx&\geq 0.
\end{array}
\qquad
\begin{array}{rrll}
(P_t)\max  & \multicolumn{2}{l}{\sum_{i\in [d],j\in [n]}x_{ij}r_{ij}a_{ij}} \\
\text{s.t.}&  \sum_{j\in [n]}x_{ij}&\leq B_i^t &\forall i \in [d+1] \\
& \sum_{i\in [d+1]}x_{ij}&= \E[Z_j(t)] &\forall j \in [n] \\
& \vx&\geq 0.
\end{array}
\end{equation}

Recall that $B^t$ represents \onl's budget with $t$ periods to go.
We solve $(P_t)$ in \cref{eq:coupled_lp_matching} and obtain an optimizer $\Xt$.
If $\Jt=j$, let $K\in\argmax\crl{\Xt_{i,j}:i\in [d+1]}$ be the maximal entry, breaking ties arbitrarily, then match $j$ to $K$.
The resulting policy is presented in \cref{alg:fluid_matching}.
Observe that, matching an agent to  $K=d+1$ (fictitious resource) is equivalent to rejecting him.

\begin{algorithm}
\caption{Fluid Bayes Selector For Online Matching}
\label{alg:fluid_matching}
\begin{algorithmic}[1]
\Require Access to solutions $\Xt$ of $(P_t)$ in \cref{eq:coupled_lp_matching}.
\Ensure Sequence of decisions for \onl.
\State Set $B^\T$ as the given initial budget levels
\For{$t=\T,\ldots,1$}
	\State Observe arrival $\Jt=j$ and let $K\gets \argmax\crl{\Xt_{ij}:i\in [d+1]}$, breaking ties arbitrarily.
	\State Match $\Jt$ to $K$.  
	\State Update $B^{t-1}_i\gets B_i^t$ for $i\neq K$ and $B^{t-1}_K \gets B^{t}_K -1$.
\EndFor
\end{algorithmic}
\end{algorithm}

\paragraph{Disagreement Sets.}
At each time $t$, matching $\Jt=j$ to $K$ requires a compensation only if \off never matches a type $j$ to $K$, i.e., $\Xts_{K,j}<1$.
On the other hand, \cref{alg:fluid_matching} picks $K$ to be the largest component, hence we should have $\Xt_{K,j}>>1$ (precisely, $\Xt_{K,j}\geq \frac{\E[Z_j(t)]}{d+1}$).
More formally, the constraint $\sum_{i\in [d+1]}x_{ij}=\E[Z_j(t)]$ in \cref{eq:coupled_lp_matching} and the definition of $S_j$ imply $\Xt_{K,j}\geq \E[Z_j(t)]/(\abs{S_j}+1)$.
We conclude that, if matching to $K$ is not satisfying (see \cref{def:satisfying}), it must be that $\norm{\Xt-\Xts}_\infty> \E[Z_j(t)]/(\abs{S_j}+1)$. 
\cref{prop:lipschitz_matching} below characterizes exactly this deviation.

Observe that, in \cref{eq:off_matching_problem,eq:coupled_lp_matching}, the matrix $A$ appears only on the objective function; this is not the usual LP formulation for this problem, but it allows us to obtain the following result.
We remark that not only we have a Lipschitz property, but the Lipschitz constant is exactly $1$.
We present the proof of \cref{prop:lipschitz_matching} in \cref{sec:proof_prop}. 

\begin{proposition}[Lipschitz Property for Matching]\label[proposition]{prop:lipschitz_matching}
Take any $z^1,z^2\in\Rp^d$ and $b\in\Rp^d$.
If $\vx^1$ is a solution of $P[z^1,b]$, then there exists $\vx^2$ solving $P[z^2,b]$ such that $\norm{\vx^1-\vx^2}_\infty\leq \norm{z^1-z^2}_1$.
\end{proposition}

From here, the proof of \cref{theo:matching} is applying the Compensated Coupling (\cref{lem:coupling}) and \cref{cor:bayes} in the same way as we did in \cref{sec:regret_rm}, hence we omit it.

\subsection{Online Stochastic Matching}
\label{sec:stoch_matching}

A classical problem that fits naturally into the above framework is that of online bipartite matching problem with stochastic inputs~\cite{manshadi2012online}.
The reader unfamiliar with the problem can find the details of the setup in \cref{sec:details_matching}.
For this setting, the bound obtained via compensated coupling surprisingly holds with equality:
\begin{lemma}\label[lemma]{lemm:matching_equal}
For the stochastic online bipartite matching, given an online policy, if $U^t$ denotes the node matched at time $t$ by \onl and $S^t$ the available nodes, then
\[
\voff-\von = \sum_{t\in[T]}\Ins{Q(t,U^t,S^t)}.
\]
\end{lemma}
Based on this, it is tempting to conjecture that the Bayes Selector does in fact lead to an optimal policy for this setting. 
This however is not the case, although showing this is surprisingly subtle; in Appendix~\ref{sec:details_matching}, we discuss this in more detail. Moreover, it is known that this problem cannot admit an expected regret that has better than linear scaling with $T$ (in particular, \citep{manshadi2012online} proves a constant upper bound on the competitive ratio for this setting).
That said, the strength of the above bound suggests that the Bayes selector may have strong approximation guarantees -- showing this remains an open problem.

\section{Regret Guarantees for Online Allocation}
\label{sec:general_allocation}

We now give the algorithm and analysis for the general Online Allocation problem defined in \cref{sec:def_allocation}.
As before, let us introduce a fictitious resource $i=d+1$ with initial capacity $B_{d+1}=T$, zero rewards ($r_{\crl{d+1}j}=0$ for all $j\in [n]$) and such that $\crl{d+1}\in S_j$ for all $j\in [n]$.
Now we can assume w.l.o.g.\ that each agent gets assigned a bundle.
Finally, for a bundle $s$, we denote $a_{is}\in\N$ the number of times the resource $i$ appears in $s$ (recall that bundles are multisets).

Given resource availability $B\in\N^{d+1}$ and total arrivals $Z\in\N^n$, we can formulate the coupled problems for \off and \onl as follows, where the variable $x_{sj}$ denotes the number of times a bundle $s\in S_j$ is allocated to a type $j$. 
\begin{equation}\label{eq:coupled_lp_bundle}
\begin{array}{rrll}
(P_t^\star)\max&  \multicolumn{2}{l}{\sum_{j\in [n],s\in S_j}x_{sj}r_{sj} }\\
\text{s.t.}&  \sum_{j\in [n],s\in S_{j}}a_{is}x_{sj}&\leq B_i^t &\forall i\in [d+1] \\
& \sum_{s\in S_j}x_{sj}&=Z_j(t) &\forall j\in [n] \\
& \vx&\geq 0.
\end{array}
\qquad
\begin{array}{rrll}
(P_t)\max&  \multicolumn{2}{l}{\sum_{j\in [n],s\in S_j}x_{sj}r_{sj} }\\
\text{s.t.}&  \sum_{j\in [n],s\in S_{j}}a_{is}x_{sj}&\leq B_i^t &\forall i\in [d+1] \\
& \sum_{s\in S_j}x_{sj}&=\E[Z_j(t)] &\forall j\in [n] \\
& \vx&\geq 0.
\end{array}
\end{equation}

We assume that the  process $Z(t)$ satisfies the All-Time Deviation bound (see \cref{def:all_time}) w.r.t.\ some norm $\mu$ and parameters $\kappa_j=(d+1)\kappa$, where $\kappa=\kappa_\mu(A)$ depends only on $A$ and $\mu$.
This condition can be restated as follows.
For every $j\in [n]$, there are constants $c_j\geq 0$ and naturals $\tau_j$ such that
\begin{equation}\label{eq:tail_zgap}
\Pr\brk*{\norm{Z(t)-\E(Z(t))}_\mu \geq \frac{\E[Z_j(t)]}{\kappa(\abs{S_j}+1)}} \leq \frac{c_j}{t^2} \quad \forall t >\tau_j.
\end{equation}

We present the resulting policy in \cref{alg:fluid_bundle} with its guarantee in \cref{theo:bundle}.
We remark that the constant $\kappa$ depends only on the constraint matrix defining the LP in \cref{eq:coupled_lp_bundle}, i.e., it does depend on the choices of bundles $S_j$, but it is independent of $T$ and $B$.

\begin{algorithm}
\caption{Fluid Bayes Selector For Online Allocation}
\label{alg:fluid_bundle}
\begin{algorithmic}[1]
\Require Access to solutions $\Xt$ of $(P_t)$ in \cref{eq:coupled_lp_bundle}.
\Ensure Sequence of decisions for \onl.
\State Set $B^\T$ as the given initial budget levels
\For{$t=\T,\ldots,1$}
	\State Observe arrival $\Jt=j$ and let $K\gets \argmax\crl{\Xt_{sj}:s\in S_j}$, breaking ties arbitrarily.
	\State If it is not feasible to assign bundle $K$, then reject. Otherwise assign $K$ to $\Jt$.  
	\State Update $B^{t-1}_i\gets B_i^t$ for $i\not\in K$ and $B^{t-1}_i \gets B^{t}_i -a_{iK}$ for $i\in K$.
\EndFor
\end{algorithmic}
\end{algorithm}

\begin{theorem}\label{theo:bundle}
For the Online Allocation problem, there exists a constant $\kappa$ that depends on $(S_j:j\in [n])$ only such that, if the arrival process satisfies the conditions in \cref{eq:tail_zgap}, then the expected regret of the Fluid Bayes Selector (\cref{alg:fluid_bundle}) is at most $\rmax\sum_{j\in [n]}p_j(c_j+\tau_j)$, where $p_j$ is an upper bound on $\Pr[\Jt=j]$.
\end{theorem}

The proof of \cref{theo:bundle} is analogous to that of \cref{theo:reg_general}, hence we omit it and provide here only the key steps.
Recall that, for request $j$, since we include the fictitious item, there are are $\abs{S_j}+1$ possible bundles.
Crucially, incorrect allocation of $s$ to $j$  necessitates $\Xt_{sj}\geq \E[Z_j(t)]/(\abs{S_j}+1)$ (because \cref{alg:fluid_bundle} takes the maximum entry)  and $\Xts_j< 1$ (\off never allocates $s$ to $j$). 
By the Lipschitz property of LPs (see \cref{prop:lipschitz}), this event requires a large deviation of $Z(t)$ w.r.t. its mean, which has low probability.
More formally, the disagreement sets (\cref{def:disagreement}) are $Q(t,b)=\crl{\omega\in\Omega: \Xt_{sj}\geq \E[Z_j(t)]/(\abs{S_j}+1) \text{ and } \Xts_j< 1}$.
By the Lipschitz property, $Q(t,b)\subseteq \crl{\omega\in\Omega: \norm{Z(t)-\E(Z(t))}_\mu \geq \frac{\E[Z_j(t)]}{\kappa(\abs{S_j}+1)}}$.
The probability of this last event is bounded by \cref{eq:tail_zgap}, hence the Compensated Coupling concludes the proof.

\section{Numerical Experiments}
\label{sec:numerics}

The theoretical results we have presented, together with known lower bounds for previous algorithms, show that our approach outperforms existing heuristics for Online Packing and Online Matching problems.
We now re-emphasize these results via simulation with synthetic data, which demonstrates both the sub-optimality of existing heuristics (in terms of expected regret which scales with $T$), as well as the fact that the Bayes selector has constant expected regret. 

We run experiments for both Online Packing and Online Matching with multinomial arrivals.
For each problem we consider two instances, i.e., two sets of parameters $(r,A,p)$, then we scale each instance to obtain a family of ever larger systems.
For each scaling we run 100 simulations.
In conclusion, we run four sets of parameters (two for packing and two for matching), each scaled to generate many systems.
The code for all the algorithms can be found in \url{https://github.com/albvera/bayes_selector}.

\subsection{Online Packing}

\begin{table}[ht]
\centering
\begin{tabular}{lllllll}
                                & \multicolumn{6}{c}{Type $j$}      \\
\multicolumn{1}{l|}{}           & 1   & 2   & 3   & 4   & 5   & 6   \\ \hline
\multicolumn{1}{l|}{Resource $i=1$} & 1   & 1   & 0   & 0   & 1   & 1   \\
\multicolumn{1}{l|}{Resource $i=2$} & 0   & 0   & 1   & 1   & 1   & 1   \\
$p_j$                           & 0.2 & 0.2 & 0.2 & 0.2 & 0.1 & 0.1 \\
$r_j$                           & 10  & 6   & 10  & 5   & 9   & 8  
\end{tabular}
\caption{Parameters for the first Online Packing instance.
Coordinates $(i,j)$ represent the consumption $A_{ij}$. 
}\label{tab:packing_one}
\end{table}

We compare the Bayes Selector against three policies:
(i) Static Randomized (SR) is the first known policy with regret guarantees, it is based on solving the fluid LP once and using the solution as a randomized acceptance rule~\citep{talluri2006theory}.
(ii) Re-solve and Randomize is based on re-solving the fluid LP at each time and using the solution as a randomized acceptance rule \citep{jasin2012}.
(iii) Infrequent Re-solve with Thresholding (IRT) is based on re-solving the fluid LP at carefully chosen times, specifically at times $\crl{T^{(5/6)^u}:u=0,1,..., \log\log(T)/\log(6/5)}$, then either randomize or threshold depending on the value of the solution \citep{wang_resolve}.

Our first instance has $d=2$ resources and $n=6$ agent types.
Types $j\in\crl{1,2}$ require one unit of resource $i=1$, types $j\in\crl{3,4}$ require one unit of $i=2$, and types $j\in\crl{5,6}$ require one unit of each resource.
All the parameters are presented in \cref{tab:packing_one}. 
We consider a base system with capacities $B_1=B_2=40$ and horizon $T=200$.
The base system is chosen such that the problem is near dual-degenerate (which is the regime where heuristics based on the fluid benchmark are known to have poor performance; see \cref{prop:bad_fluid}).
Finally, for a scaling $k\in\N$, the $k$-th system has capacities $kB$ and horizon $(k+k^{0.7})T$.
We remark that traditionally the horizon is scaled as $kT$, but we chose this slightly different scaling to emphasize that our result does not depend on the specific way the system is scaled.

The results for the first instance are summarized in \cref{fig:packing_avg}, where we also present a log-log plot which allows better to appreciate how the regret grows. 
Static Randomized has the worst performance in our study; indeed, we do not include it in the plot since it is orders of magnitude higher.
We note that not only the Bayes Selector outperforms previous methods, but the regret is very small (both in average and sample-path wise), specially in comparison with the overall reward which grows linearly with $k$, i.e. $\voff = \Omega(k)$ (in expectation and w.h.p.).

\begin{figure}[ht]
    \centering
    \includegraphics[scale=0.68]{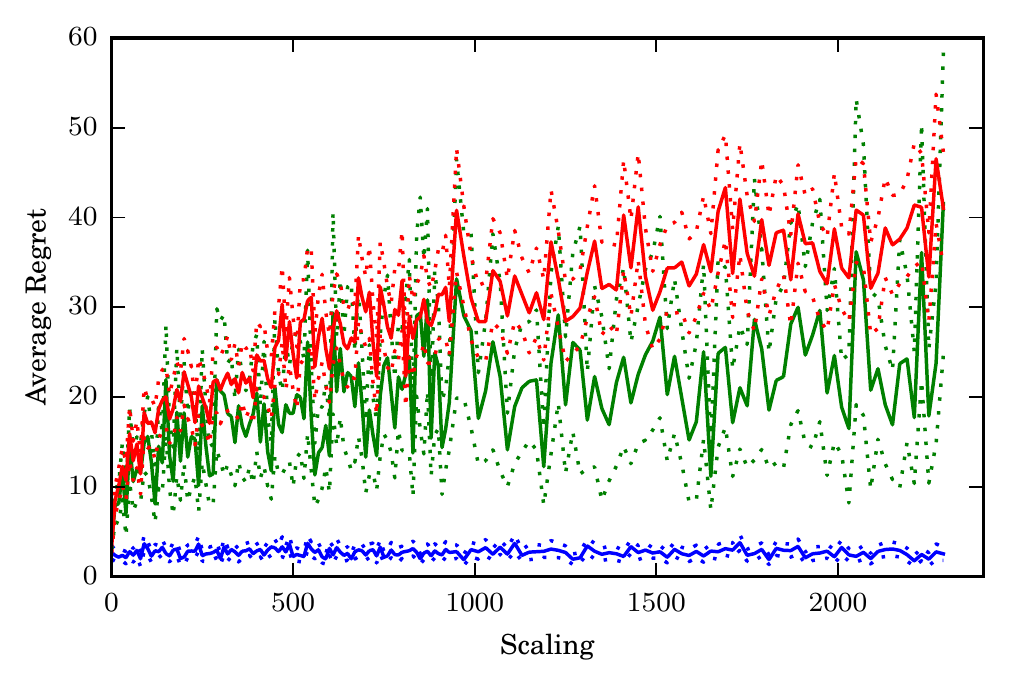}%
    \includegraphics[scale=0.68]{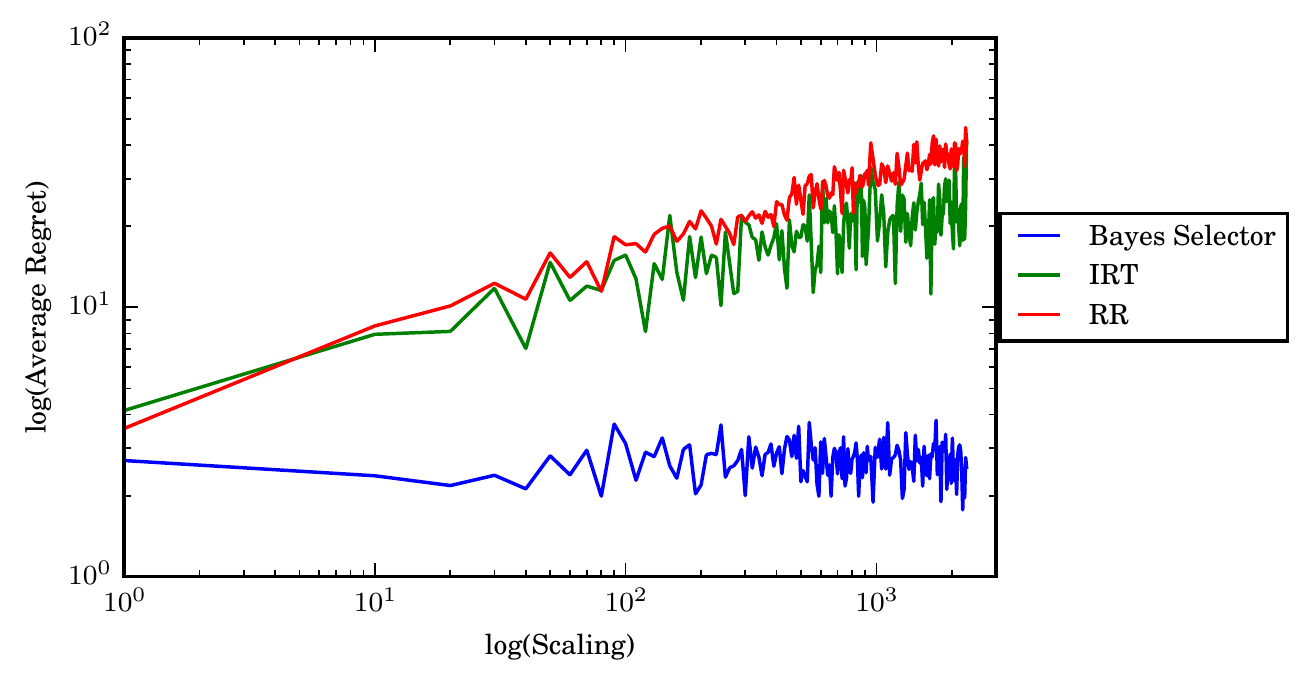}
    \caption{Average regret of different policies for Online Packing in the first instance.
    We present a plot on the left and a log-log plot on the right.
    We run the Bayes Selector, Infrequent Re-solve with Thresholding (IRT) \citep{wang_resolve}, Re-solve and Randomize (RR) \citep{jasin2012}, and Static Randomized (SR) \citep{talluri2006theory} (this last one is not reported because its high regret distorts the figures).
    The plot shows the regret incurred by the policies versus the offline optimum, for different scalings.
    Dotted lines represent 90\% confidence intervals.
    }
    \label{fig:packing_avg}
\end{figure}

The second instance has $n=15$ agent types and $d=20$ resources, the specific parameters are presented in \cref{tab:packing_two} in \cref{appen:numerics} and were generated randomly.
We take a base system with horizon $T=50$ and capacities $B_i=10$ for all $i\in[20]$, then the $k$-th system has horizon $kT$ and capacities $kB$.
The performance of different algorithms is presented in \cref{fig:packing_two}.
We notice that this instance is not degenerate and we are scaling linearly, hence all the algorithms except Static Randomize (which we again omit from the plots) are known to achieve constant regret.
Nevertheless, we observe that the Bayes Selector has the best performance by a large margin.

\begin{figure}[ht]
    \centering
    \includegraphics[scale=0.9]{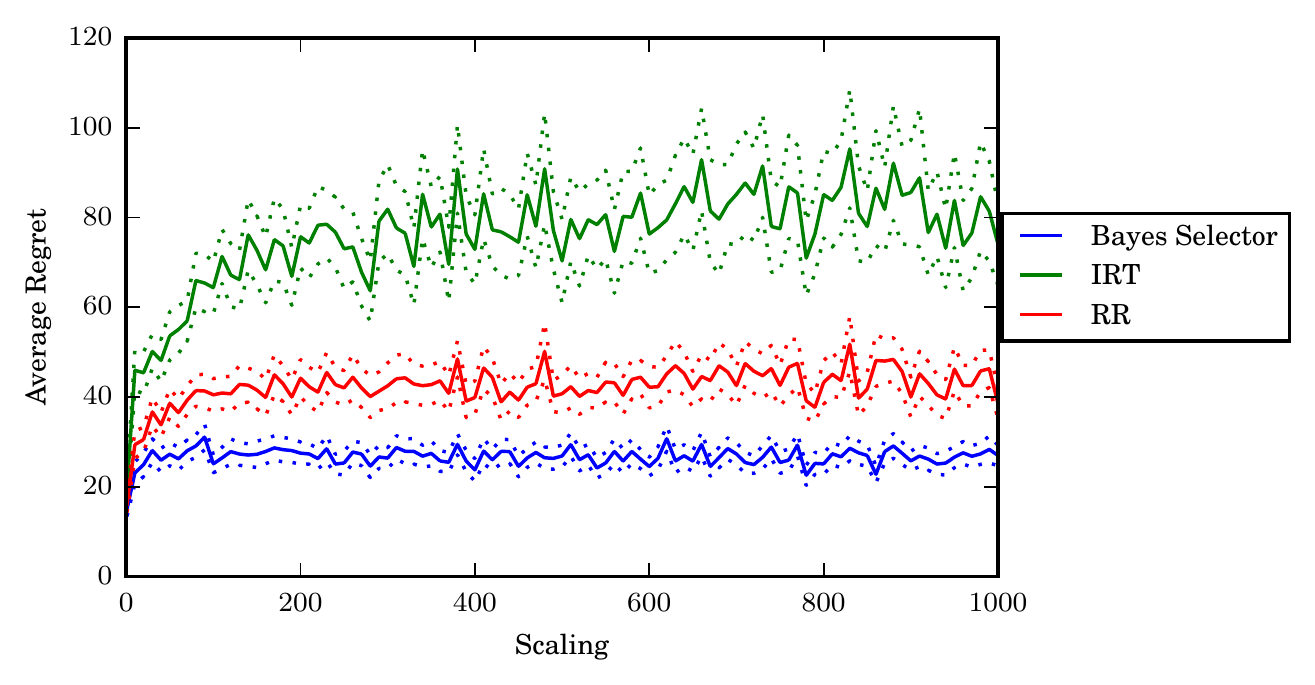}
    \caption{Average regret of different policies for Online Packing in the second instance.
    We run the Bayes Selector, Infrequent Re-solve with Thresholding (IRT) \citep{wang_resolve}, Re-solve and Randomize (RR) \citep{jasin2012}, and Static Randomized (SR) \citep{talluri2006theory} (this last one is not reported because its high regret distorts the figures).
    The plot shows the regret incurred by the policies versus the offline optimum, for different scalings.
    Dotted lines represent 90\% confidence intervals.
    }
    \label{fig:packing_two}
\end{figure}

\subsection{Online Matching}

As we mentioned in \cref{sec:matching_type}, our problem corresponds to stochastic matching with edge weights.
There has been previous work studying constant factor approximations for worst-case distributions.
In particular, the state of the art is a $0.705$ competitive ratio \citep{matching_new}, while a previous algorithm achieved a $0.667$ competitive ratio \citep{matching_old}.
Both algorithms are impractical since they require a sampling procedure over $\text{poly}(T\cdot\max_{i\in[d]}B_i)$ many matchings.
To the best of our knowledge, the best guarantee of a practical algorithm is a $1-1/e\approx 0.63$ competitive ratio and is achieved by the base algorithm in \citep{matching_old} (that when built upon achieves the $0.667$ guarantee).
We therefore benchmark against this algorithm, which we call ``Competitive".

Competitive is based on solving a big LP once (it has $\Omega(T\cdot\max_{i\in [d]}B_i)$ variables) and using the solution as a probabilistic allocation rule.
We also compare against a contemporaneous algorithm, called Marginal Allocation, that is based on bid-prices \citep{wang_scheduling}.
Marginal Allocation uses approximate dynamic programming to obtain the marginal benefit of a matching, then uses this marginal value as a bid price so that, if the reward exceeds it, then we match the request.
We give further details for both Marginal Allocation and Competitive in \cref{appen:numerics}.

The first instance we consider has $d=2$ resources and $n=6$ agent types.
The specific parameters are presented in \cref{tab:matching_one}, where reward $r_{ij}=0$ implies that type $j$ cannot be matched to that resource $i$, i.e., $A_{ij}=0$.
We consider a base system with horizon $T=20$ and capacities $B=(4,5)'$ and then scale it linearly so that the $k$-th system has horizon $kT$ and capacities $kB$. 
Our second instance has $d=6$ resources and $n=10$ agent types, the specific parameters are presented in \cref{tab:matching_two} in \cref{appen:numerics}.
We consider a base system with horizon $T=200$ and capacities $B=(40,50,40,30,20,40)'$ and then scale it linearly so that the $k$-th system has horizon $kT$ and capacities $kB$.

\begin{figure}[ht]
    \centering
    \includegraphics[scale=0.68]{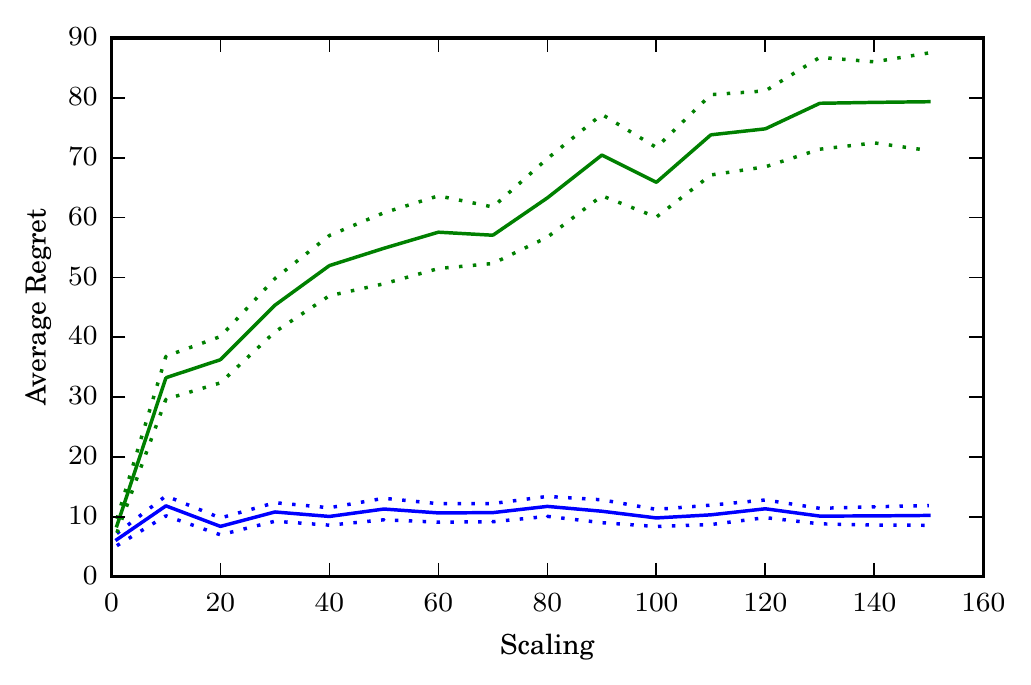}%
    \includegraphics[scale=0.68]{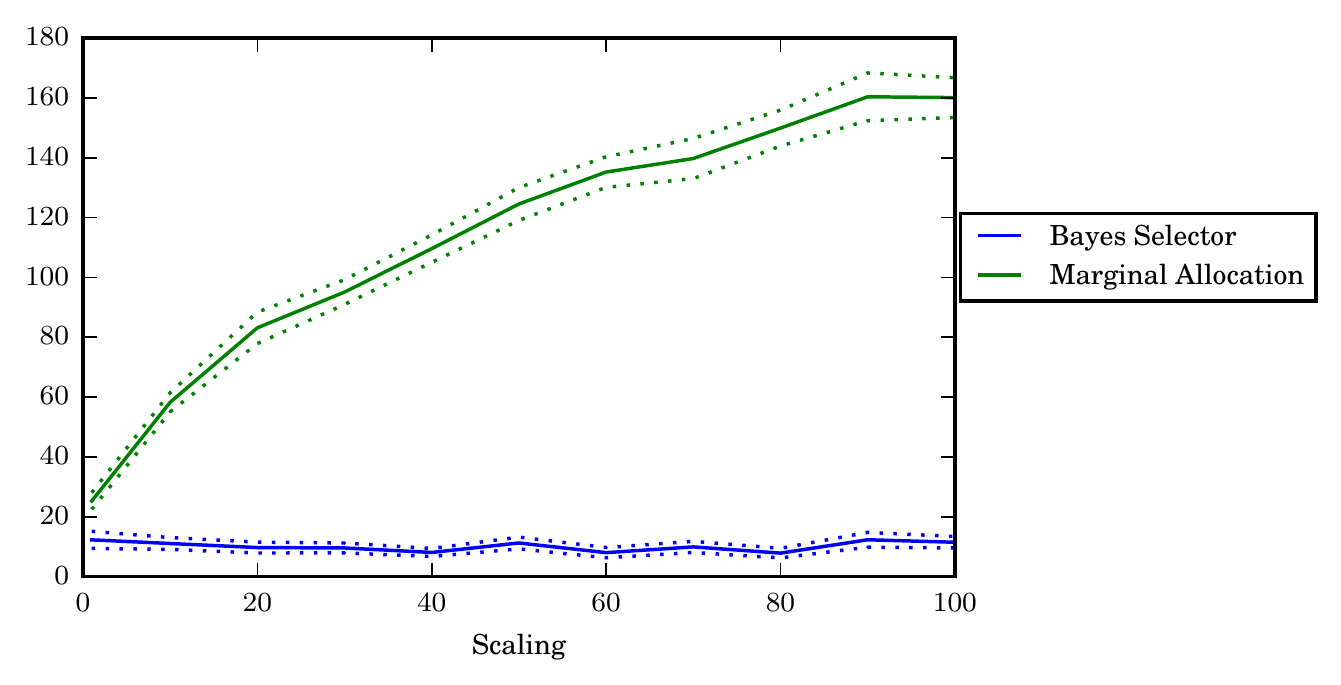}
    \caption{Average regret of different policies for Online Matching.
    First instance on the left and second on the right.
    We run the Bayes Selector, Marginal Allocation \citep{wang_scheduling}, and Competitive \citep{matching_old} (this last one is not reported because its high regret distorts the figures).
    The plot shows the regret incurred by the policies versus the offline optimum, for different scalings.
    }
    \label{fig:matching}
\end{figure}

\begin{table}[ht]
\centering
\begin{tabular}{lllllll}
                               & \multicolumn{6}{c}{Type $j$}      \\
\multicolumn{1}{l|}{}          & 1   & 2   & 3   & 4   & 5   & 6   \\ \hline
\multicolumn{1}{l|}{Resource 1}  & 10   & 6   &  0  & 0   & 9 & 8 \\
\multicolumn{1}{l|}{Resource 2} & 0  & 0   & 5  & 10   & 20   & 20   \\
\multicolumn{1}{l|}{$p_j$}     & 0.2 & 0.2 & 0.2 & 0.2 & 0.1 & 0.1
\end{tabular}
\caption{Parameters used for the first Online Matching instance.
Coordinates $(i,j)$ represent the reward $r_{ij}$ and $r_{ij}=0$ implies that it is not possible to match $i$ to $j$.}\label{tab:matching_one}
\end{table}

The results are presented in \cref{fig:matching}.
We do not include the regret of Competitive, since it is so high that it distorts the plots (it starts at 80 times the regret of the other algorithms and then grows linearly with $k$).
We can confirm that the Bayes Selector has constant regret and, additionally, offers the best performance.
Marginal Allocation offers a much better performance than Competitive, but its regret still grows and seems to scale as $\Omega(\sqrt{T})$.

\section{Conclusions}

We reiterate that our contributions in this paper are both to develop new online policies that achieve constant regret for a large class of online resource allocation problems, and also, a new technique for analyzing online decision-making heuristics.

Our work herein has developed a new technical tool---the Compensated Coupling---for analyzing online decision making policies with respect to offline benchmarks.
In short, the main insight is that, through the use of compensations, we can couple \off's state to that of \onl on every sample path.
This simplifies the analysis of online policies since, in contrast to existing approaches, we do not need to track the complicated offline process.

Next, we presented a general class of problems, which we referred to as Online Allocation, wherein different agents request different bundles of resources.
This problem captures, among others, Online Packing and Online Matching.
For all of these problems we present a tractable policy, the Bayes Selector, based on re-solving an LP, that achieves constant regret.

Our analysis is based on the Compensated Coupling and, thanks to its versatility, we can accommodate a large class of arrival processes including: correlated processes, heavy tailed, and the classical Poisson and Multinomial (i.i.d.).

Although we instantiate the Bayes Selector for Online Allocation, we defined it for general MDPs; we hope this policy is useful for other types of problems too.
We remark two properties of the Bayes Selector: (i) it works on interpretable quantities, namely the estimation of disagreement probabilities $\hat q$, and (ii) it is amenable to simulation, since $\hat q$ can be estimated by running offline trajectories.
We therefore think that a promising avenue for further research is to apply this policy to other problems using modern estimation techniques.

The assumption of finite types of agents is well founded in revenue management problems, but there are settings where the number of types could be very large or even continuous.
Based on reported numerical results~\citep{itai_secretary}, the Bayes Selector appears to have good performance even in this setting.
An interesting problem is to obtain parametric guarantees (not worst case) in the case where the number of types is very large or continuous.

\section*{Acknowledgement}
We gratefully acknowledge support from the ARL under grant W911NF-17-1-0094, and the NSF under grant DMS-1839346 and ECCS-1847393. A preliminary version of this work was presented at SIGMETRICS 2019 and accompanied by a one-page extended abstract \citep{vera2019bayesian}. We thank the anonymous reviewers for their helpful feedback.


\bibliographystyle{plain} 



\bibliography{biblio}


%
%
\begin{APPENDICES}
\section{The Fluid Benchmark}\label[appendix]{sec:proofs}
\begin{proofof}{\cref{prop:bad_fluid}}
To build intuition, we start with a description of dual degeneracy  for the online knapsack problem with budget $B\leq\T$.
We assume w.l.o.g.\ $r_1\geq r_2\geq \ldots \geq r_n$ and denote $Z=Z(T)$.
The primal and dual are given by
\begin{align*}
\begin{array}{rrl}
(P[Z])\, \max  & r' x& \\
\text{s.t.}&  \sum_{j\in [n]}x_j&\leq B \\
& x&\leq Z \\
& x&\geq 0,
\end{array}
\qquad \qquad
\begin{array}{rrl}
(D[Z])\, \min  & \alpha B + \beta' Z& \\
\text{s.t.}&  \alpha + \beta_j &\geq r_j \, \forall j  \\
& \alpha &\geq 0\\
& \beta &\geq 0.
\end{array}
\end{align*}

Let us denote $\mu\defeq \E[Z]$.
If the fluid $(P[\mu])$ is degenerate, then we have $n+1$ active constraints.
It is straightforward to conclude that there must be an index $\jstar$ such that $\sum_{j\leq \jstar}\E[Z_j] = B$.
The fluid solution is thus $x_j= \E[Z_j]$ for $j\leq \jstar$ and $x_j=0$ for $j>\jstar$.
We can construct two dual solutions as follows.
Let $\alpha^1=r_\jstar$ and $\alpha^2 = r_{\jstar+1}$, these correspond to the shadow prices for alternative budgets $B-\varepsilon$ and $B+\varepsilon$ respectively.
The corresponding variables $\beta^1,\beta^2$ are given by $\beta^k_j =(r_j-\alpha^k)_+$ for $k=1,2$.
Intuitively, the fluid is indifferent between these two dual bases, but, given a realization of $Z$, \off will prefer one over the other; this causes a discrepancy between the expectations.

Now we turn to the case of any packing problem, the assumption is that we are given two optimal dual solutions $(\alpha^k,\beta^k)$, with $\beta^1\neq\beta^2$.
The dual is a minimization problem and $(\alpha^k,\beta^k)$ are always dual feasible, thus defining $\beta\defeq \beta^1-\beta^2$ and $\alpha\defeq \alpha^1-\alpha^2$,
\begin{align*}
v(D[Z]) \leq \min_{k=1,2}\crl{B'\alpha^k+Z'\beta^k}
= (B'\alpha^1+Z'\beta^1)\In{ B'\alpha +Z'\beta <0} + (B'\alpha^2+Z'\beta^2)\In{B'\alpha +Z'\beta \geq 0}.
\end{align*}

The rest of the proof is reasoning that interchanging expectations $\E[\min_{k=1,2}\crl{B'\alpha^k+Z'\beta^k}]$ for $\min_{k=1,2}\crl{B'\alpha^k+\E[Z]'\beta^k}$ induces a $\Omega(\sqrt{\T})$ error.

Since the two dual solutions have the same dual value,  $B'\alpha^1+\mu'\beta^1 = B'\alpha^2+\mu'\beta^2$, we conclude $ B'\alpha = -\mu'\beta$.
We can use this condition to rewrite our bound as
\begin{align*}
v(P[Z]) \leq v(D[Z]) 
\leq (B'\alpha^1+Z'\beta^1)\In{ (\mu-Z)'\beta >0} + (B'\alpha^2+Z'\beta^2)\In{(\mu-Z)'\beta\leq 0}.
\end{align*}
Since $v(P[\mu]) = B'\alpha^k+\mu'\beta^k $ for $k=1,2$, we take a random convex combination to obtain 
\begin{align*}
v(P[\mu]) = ( B'\alpha^1 + \mu'\beta^1)\In{(\mu-Z)'\beta >0} +  ( B'\alpha^2+\mu'\beta^2)\In{(\mu-Z)'\beta \leq 0}.
\end{align*}
Now combine the last with our upper bound for $v(P[Z])$ and take expectations to obtain
\begin{align*}
v(P[\mu]) - \E[v(P[Z])] 
&\geq \E[(\mu-Z)'\beta^1\In{(\mu-Z)'\beta >0}] + \E[(\mu-Z)'\beta^2\In{(\mu-Z)'\beta \leq 0}]\\
&= \E[(\mu-Z)'\beta^1\In{(\mu-Z)'\beta >0}] + \E[(\mu-Z)'\beta^2(1-\In{(\mu-Z)'\beta> 0})]\\
&= \E[(\mu-Z)'\beta \In{(\mu-Z)'\beta >0}].
\end{align*}
Let us define $\xi$ as the normalized vector $Z$, i.e., $\xi\defeq \frac{1}{\sqrt{T}}(\mu-Z)$.
We conclude that
\begin{align*}
v(P[\mu]) - \E[v(P[Z])] \geq \sqrt{\T}\E[\xi'\beta \In{\xi'\beta>0}].
\end{align*}
Reducing by the standard deviation and applying the Central Limit Theorem, we arrive at a half-normal (also known as folded normal), which has constant expectation.
This concludes the desired result.
\end{proofof}

\section{Additional Details and Proofs} 

\subsection{Poisson Process in Discrete Periods} \label[appendix]{sec:poisson}

We explain how a continuous time Poisson process can be reduced to our setting.
We are given a time horizon $T$, where time $t\in [0,T]$ still denotes time to go
and, according to an exponential clock, arrivals occur at some times $t_1>t_2>\ldots >t_N \in [0,T]$, where $N$ is random and corresponds to the total number of arrivals, i.e., $N=\sum_{j\in [n]}Z_j(T)$.

Treating times $t_k$ as periods, there is one arrival per period.
Observe that \off knows $N$, therefore his Bellman Equation is well defined.
\onl acts on these discrete periods, i.e., he is event-driven, thus making at most $N$ decisions.
Finally, we note that, at some time $t_k$, $\E[Z_j(t_k)]=\lambda_j t_k$ if the process is homogeneous or $\E[Z_j(t_k)] = \int_0^{t_k}\lambda_j(t)dt$ if the process is non-homogeneous
In conclusion, \onl can compute all the required expectations without knowing $N$, but rather the knowledge of $t_k$ and $\lambda(\cdot)$ is enough.

\subsection{Bayes Selector Based on Marginal Compensations}\label[appendix]{sec:loss}
A somewhat more powerful oracle is one which, for every time $t$, state $s$ and action $a$, returns estimates of the \emph{marginal compensation} $\Rl(t,a,s)\cdot\Ins{Q(t,a,s)}$.
This suggests a stronger form of the Bayes selector based on marginal compensations, as summarized in \cref{alg:loss}.

The following result follows directly from \cref{lem:coupling} and gives a performance guarantee for this algorithm.

\begin{algorithm}
\caption{Marginal-Compensation Bayes Selector}
\label{alg:loss}
\begin{algorithmic}[1]
\Require Access to over-estimates $\hat l(t,a,s)$ of the expected compensation, i.e.,  $\hat l(t,a,s)\geq \E[\Rl(t,a,s)\Ins{Q(t,a,s)}]$
\Ensure Sequence of decisions for \onl.
\State Set $S^\T$ as the given initial state
\For{$t=\T,\ldots,1$}
	\State Observe arrival $\Jt$, and take any action that minimizes marginal compensation, i.e., $a\in\argmin\crl{\hat l(t,a,S^t): a\in\Ac}$.
	\State Update state $S^{t-1}\gets \Tr(a,S^t,\Jt)$.
\EndFor
\end{algorithmic}
\end{algorithm}

\begin{corollary}[Regret Of Marginal-Compensation Bayes Selector]
Consider \cref{alg:loss} with overestimates $\hat l(t,a,s)$,  i.e.,  $\hat l(t,a,s)\geq \E[\Rl(t,a,s)\Ins{Q(t,a,s)}]$.
If $A^t$ denotes the policy's action at time $t$, then 
\begin{align*}
\E[\reg]\leq 
\sum_t \E[\hat l(t,A^t,S^t)].
\end{align*}
\end{corollary}

\subsection{Multi-Secretary Problem}\label{sec:app_multi}

\begin{proofof}{\cref{theo:secretary}}
Assume w.l.o.g.\ that $r_1\geq r_2\geq\ldots\geq r_n$.
This one dimensional version can be written as follows.
\begin{equation*}
\begin{array}{rrl}
(P_t^\star)\; \max & r'x &\hspace{-0.2cm}
 \\
\text{s.t.}& \sum_{j\in [n]}x_j &\leq B^t \\
& x_j&\leq Z_j(t)\; \forall j \\
& x&\geq 0.
\end{array}
\qquad\qquad
\begin{array}{rrl}
(P_t)\; \max  & r'x& \hspace{-0.2cm}
\\
\text{s.t.}&  \sum_{j\in [n]}x_j&\leq B^t \\
& x_j&\leq tp_j\; \forall j \\
& x&\geq 0.
\end{array}
\end{equation*}
The optimal solution to $(P_t^\star)$ is to sort all the arrivals by reward and pick the top ones.
The solution to $(P_t)$ is similar, except that it can be fractional; we saturate the variable $x_1$ to $tp_1$, then $x_2$ to $tp_2$, and continue as long as  $\sum_{i\leq j} tp_i \leq B^t$ for some $j$.  
Define the probability of `arrival $j$ or better' by $\pb_j\defeq \sum_{i\leq j}p_{i}$.
Observe that we can saturate all variables $1,\ldots,j$ iff $t\pb_j \leq B^t$.
The solution to $(P_t)$ is therefore to pick the largest $j$ such that $t\pb_j\leq B^t$, then make $\Xt_i=tp_i$ for $i\leq j$ and $\Xt_{j+1}=B^t-t\pb_j$.
When we round this solution according to \cref{alg:fluid}, we arrive at the following policy:
First, if $B^t=0$, end the process.
Second (assuming $B^t \geq 1$), always accept class $j=1$.
Third (assuming $B^t\geq 1$), if class $j>1$ arrives, accept if $B^t/t\geq \pb_{j}-p_j/2$ and reject if $B^t/t<\pb_{j}-p_j/2$.

Recall that $q(t,b)$ is the probability that \off is not satisfied with \onl's action at time $t$ if the budget is $b$.
We denote $q_j(t,b)$ as the probability conditioned on $\Jt=j$.
Our aim in the rest of the section is to show that $q_j(t,b)$ is summable over $t$.

As we observed before:  (1) \off is not satisfied rejecting a class $j$ iff he accepts all the future arrivals type $j$, i.e., $\Xts_j>Z_j(t)-1$.
(2) \off is not satisfied accepting class $j$ iff he rejects all future type $j$ arrivals, i.e., $\Xts_j<1$.
We use the following standard Chernoff bound in \citep[Theorem 1.1]{dubhashi}.
For any $\alpha\in [0,1]$, if $X\sim\text{Bin}(t,\alpha)$:
\begin{align}
\label{eq:hoef_bino}
\Pr[X-\E[X]\leq -t\varepsilon] \leq e^{-2\varepsilon^2t}\;, \; 
\Pr[X-\E[X]\geq t\varepsilon] \leq e^{-2\varepsilon^2t}.
\end{align}

We now bound the disagreement probabilities $q_j(t,B^t)$.
Take $j$ rejected by \onl, i.e., it must be that $j>1$ and $B^t/t<\pb_j-p_j/2$.
Since we are rejecting, a compensation is paid only when condition (1) applies, thus $\Xts_j=Z_j(t)$.
By the structure of \off's solution, all classes $j'\leq j$ are accepted in the last $t$ rounds, i.e., it must be that $\Xts_{j'}=Z(t)_{j'}$ for all $j'\leq j$.
We must be in the event $\sum_{j'\leq j} Z(t)_{j'}\leq B^t$.
We know that $\sum_{j'\leq j} Z(t)_{j'}\sim \Bin(t,\pb_j)$.
Since $B^t/t<\pb_j-p_j/2$, the probability of error is:
\begin{align*}
q_j(t,B^t)
\leq
\Pr\brk*{\sum_{j'\leq j} Z(t)_{j'}\leq B^t} 
= 
\Pr[\Bin(t,\pb_j)\leq B^t]
\leq \Pr[\Bin(t,\pb_j)\leq t\pb_{j}-tp_j/2].
\end{align*}
Using \cref{eq:hoef_bino}, it follows that $q_j(t,B^t)\leq e^{-p_j^2t/2}$.

Now let us consider when $j$ is accepted by \onl.
A compensation is paid only when $j>1$ and condition (2) applies, thus $\Xts_j=0$.
Again, by the structure of $\Xts$, necessarily $\Xts_{j'}=0$ for $j'\geq j$.
Therefore we must be in the event $\sum_{j'< j} Z(t)_{j'}\geq B^t$.
Recall that $j$ is accepted iff $B^t/t\geq \pb_j-p_j/2=\pb_{j-1}+p_j/2$, thus
\begin{align*}
q_j(t,B^t)\leq
\Pr\brk*{\sum_{j'< j} Z(t)_{j'}\geq B^t} = 
\Pr[\Bin(t,\pb_{j-1})\geq B^t]
\leq \Pr[\Bin(t,\pb_{j-1})\geq t\pb_{j-1}+tp_j/2]. \end{align*}
This event is also exponentially unlikely.
Using \cref{eq:hoef_bino}, we conclude $q_j(t,B^t)\leq e^{-p_j^2t/2}$.
Overall we can bound the total compensation as:
\begin{align*}
\sum_{t\leq T} q(t,B^t) \leq \sum_{j>1}p_j\sum_{t\leq T}e^{-p_j^2t/2}
\leq \sum_{j>1}p_j\frac{2}{p_j^2}.
\end{align*}
Using compensated coupling (\cref{lem:coupling}), we get our result.
\end{proofof}

\begin{proofof}{\cref{cor:secretary}}
By \cref{cor:additive}, if $A^t$ is the action using over-estimates $\hat q$, then $\E[\reg] \leq \rmax \sum_{t\in [T]} (\E[\hat q(t,A^t,B^t)]+\Delta^t)$.
Recall that $A^t$ is chosen to minimize disagreement, hence, given the condition $\abs{q(t,a,b)-\hat q(t,a,b)}\leq \Delta^t$, we have $\hat q(t,A^t,B^t) \leq \min_{a\in\Ac} q(t,a,B^t)+\Delta^t$. 
In conclusion,
\[
\E[\reg] \leq \rmax \sum_{t\in [T]} \prn*{\E\brk*{\min_{a\in\Ac} q(t,a,B^t)}+2\Delta^t}.
\]
We prove that $\min_{a\in\Ac}q_j(t,a,b) \leq e^{-p_j^2t/2}$ for all $t\in [T],j\in [n], b\in\N$, hence the corollary follows by summing all the terms.

Let us denote $a=1$ the action accept and $a=0$ reject.
In the proof of \cref{theo:secretary} we concluded that the following are over-estimates of the disagreement probabilities $q$:
\begin{equation}\label{eq:q_secre}
\hat q_j(t,1,b) = \left\{
\begin{array}{ll}
e^{-p_j^2t/2}     & \text{ if } \frac{\Xt_j}{tp_j}\geq 1/2 \\
1     & \text{otherwise}. 
\end{array}
\right.
\quad \text{ and } \quad
\hat q_j(t,0,b) = \left\{
\begin{array}{ll}
e^{-p_j^2t/2}    & \text{ if } \frac{\Xt_j}{tp_j}< 1/2 \\
1     & \text{otherwise}. 
\end{array}
\right.
\end{equation}
Crucially, observe that the term $e^{-p_j^2t/2}$ is \emph{independent of the state} $b$.
This proves that $\sup_{b\in\N}\min\crl{q_j(t,0,b),q_j(t,1,b)}\leq e^{-p_j^2t/2}\,\forall\, t\in[T],\,\forall\, j\in\J$.
The proof is completed.
\end{proofof}

\subsection{Other Arrival Processes}\label{sec:app_arrival}

\begin{proofof}{\cref{ex:markov}}
This follows from an application of \citep[Theorem 3.1]{markov_chernoff}, which guarantees that, for some constants $c',m$ that depend on $P$ only,  
\begin{equation}\label{eq:markov_chernoff}
\Pr[\abs{Z_k(t)-\nu_kt}\geq \delta \nu_k t] \leq c' e^{-\frac{\delta^2 \nu_k t}{72m}}, \quad \forall t\in [T], \delta\in [0,1], k \in [n].
\end{equation}
To obtain \cref{eq:markov_all_time}, we fix $j\in [n]$ and use a union bound taking the worst case in \cref{eq:markov_chernoff}; we let $\nu_{\min} \defeq \min_{k\in [n]}\nu_k$, $\nu_{\max} \defeq \max_{k\in [n]}\nu_k$ and set $\delta = \nu_j/2\kappa_j\nu_{\max}$ in \cref{eq:markov_chernoff} to obtain the result.
The constants are thus $c_j = (\nu_j/2\kappa_j\nu_{\max})^2\nu_{\min}/72m$.
Finally, we mention that the constants $c'$ and $m$ are related to the spectral gap and mixing time of $P$, for details see \citep{markov_chernoff}.
\end{proofof}

\begin{proofof}{\cref{ex:heavy}}
To prove the All-Time Deviation, we use that, from the proof of \citep[Lemma~3]{multinomial_bound}, $\Pr(\abs{X-\E[X]}\geq\varepsilon\E[X])\leq 2e^{-\E[X]\varepsilon^2/4}$ is valid for any Poisson r.v. $X$ and any $\varepsilon>0$.
Now we proceed as in \cref{ex:markov}: taking $X_k=Z_k(t)$ and $\varepsilon = \frac{\E[Z_j(t)]}{2\kappa_j\E[Z_k(t)]}$ we obtain
\[
\Pr\brk*{\norm{Z(t)-\E[Z(t)]}_\infty\geq \frac{\E[Z_j(t)]}{2\kappa_j}} \leq 2\sum_{k\in [n]}e^{-\frac{\E[Z_j(t)]^2}{8\kappa_j^2\E[Z_k(t)]}}.
\]
Finally, from \cref{eq:heavy_ratio} we have $\E[Z_k(t)]\leq g(t)\E[Z_j(t)]$ and from \cref{eq:heavy_min} we have $\E[Z_j(t)]\geq g(t)f(t)\log(t)$.
From these bounds, we conclude $\Pr\brk*{\norm{Z(t)-\E[Z(t)]}_\infty\geq \frac{\E[Z_j(t)]}{2\kappa_j}} \leq 2 n e^{-f(t)\log(t)/8\kappa_j^2}$ and the existence of constants $\tau_j,c_j$ satisfying the All-Time Deviation follows.

\end{proofof}

\subsection{Proof of \texorpdfstring{\cref{prop:lipschitz_matching}}{}}\label[appendix]{sec:proof_prop}

We denote $\vx\in\R^{nd}$ the vector of the form $\vx=(x_{11},x_{21}\ldots,x_{d1},x_{12},\ldots)'$, i.e., we concatenate the components $x_{ij}$ by $j$ first. 
We can write the feasible region of $P[z,b]$ as $\crl{\vx:C\vx\leq b, D\vx\leq z, \vx\geq 0}$, where $C\in\R^{d\times nd}$ and $D\in\R^{n\times nd}$.
It follows from a slight strengthening of \cite[Theorem 2.4]{mangasarian} that $\norm{\vx^1-\vx^2}_\infty\leq \kappa\norm{z^1-z^2}_1$, where
\begin{align*}
\kappa= \sup\crl*{
\norm{v}_\infty:\norm{C'u+D'v}_1=1, \mbox{ support} \begin{pmatrix} u \\ v\end{pmatrix} \text{corresponds to linearly independent rows of } \begin{pmatrix} C \\ D\end{pmatrix} }
\end{align*}
If we study \cref{eq:off_matching_problem}, denoting $I_d$ the $d$-dimensional identity and $1_d,0_d$ $d$-dimensional row vectors of ones and zeros, we can write the matrices $C,D$ as follows.
We sketched the multipliers $u_i,v_j$ next to the rows,
\begin{align*}
C = [I_d|I_d|\ldots|I_d] 
=\begin{pmatrix} 
1 & 0&\cdots& 0 &1&\cdots\\ 
0 &1&\cdots&0&0&\cdots \\ 
\vdots &\vdots&\ddots & \vdots&\vdots &\cdots\\
 0&0&\cdots&1&0&\cdots
\end{pmatrix}
\begin{matrix}\leftarrow u_1 \\ \leftarrow u_2 \\ \vdots \\ \leftarrow u_d\end{matrix}
\end{align*}
and similarly
\begin{align*}
D = \begin{pmatrix} 1_d & 0_d& \cdots& 0_d \\ 0_d &1_d&\cdots&0_d \\ \vdots &\vdots &\ddots & \vdots \\ 0_d&0_d&\cdots&1_d\end{pmatrix}
\begin{matrix}\leftarrow v_1 \\ \leftarrow v_2 \\ \vdots \\ \leftarrow v_n\end{matrix}
\end{align*}
We have two cases: either $u_i=0$ for some $i\in [d]$ or $u_i\neq 0$ for all $i\in [d]$.
On the first case, say w.l.o.g.\ $u_1=0$ and take any $j\in [n]$. 
Observe that the constraint $\norm{C'u+D'v}_1=1$ implies (studying all the components involving $j$) $\sum_{i\in [d]}\abs{u_i+v_j}\leq 1$.
Since $u_1=0$, this reads as $\abs{v_j}+\sum_{i>1}\abs{u_i+v_j}\leq 1$, thus $\abs{v_j}\leq 1$ as desired.

For the other case we assume $u_i\neq 0$ for all $i$, hence $v_j=0$ for some $j$, since otherwise we would violate the l.i.\ restriction on the support.
Assume w.l.o.g.\ $v_1=0$ and let us study some $v_j$.
The constraint $\norm{C'u+D'v}_1=1$ implies (looking at the first $n$ components and the components involving $j$) $\sum_{i \in [d]}\abs{u_i}+\sum_{i \in [d]}\abs{v_j+u_i}\leq 1$.
By triangle inequality,
\begin{align*}
d\abs{v_j} = \abs*{\sum_{i \in [d]}(u_i+v_j) - \sum_{i \in [d]}u_i}
\leq \sum_{i \in [d]}\abs{v_j+u_i}+\sum_{i \in [d]}\abs{u_i}\leq 1.
\end{align*}
This shows $\abs{v_j}\leq 1$ and the proof is complete.  \hfill $\square$

\subsection{Additional Details for Online Stochastic Matching}\label[appendix]{sec:details_matching}

The stochastic bipartite matching is defined by a set of static nodes $U$, $\abs{U}=d$, and a random set of nodes arriving sequentially.
At each time a node $\Jt$ is chosen from a set $V$, $\abs{V}=n$, and we are given its set of neighbors in $U$.
We identify the online bipartite matching problem in our framework as follows.
The state $S^t$ encodes the available nodes from $U$, an action corresponds to matching the arrival $\Jt\in V$ to a neighbor $u\in U$ of $\Jt$ or to discard the arrival.
In the latter case we say that it is matched to $u=\varnothing$.

For a graph $G$, we denote the size of its maximum matching as $M(G)\in\N\cup\crl{0}$ and $G-(u,v)$ as the usual removal of nodes; in the case $u=\varnothing$, $G-(u,v)=G-v$.
Recall that $Q(t,a,s)$ is the event when \off is not satisfied with action $a$ and $q(t,a,s)=\Pr[Q(t,a,s)]$.
Let us fix an \onl policy and define $G_t=(L,R)$ as the bipartite graph with nodes $L=S^t$ and $R=Z(t)$, i.e., the realization of future arrivals and current state.
With the convention $\Ins{\varnothing}=0$ and $\Ins{u}=1$ for $u\in U$,
\begin{align*}
\bar Q(t,u,s) = \crl{\omega\in\Omega: M(G_t) =\Ins{u}+M(G_t-(u,\Jt))}.
\end{align*}
In words, \off is satisfied matching $\Jt$ to $u$ if the size of the maximum matching with and without that edge differs by exactly $1$. 
With this observation, a straightforward application of the compensated coupling \cref{lem:coupling} yields \cref{lemm:matching_equal}.

Finally, we provide an example for a negative result.
Despite the fact that the regret is exactly the number of disagreements and the Bayes Selector minimizes each term, it is not an optimal policy.

\begin{proposition}
The Bayes Selector is sub-optimal for stochastic online bipartite matching.
\end{proposition}
\begin{proof}
Consider an instance with static nodes $U=\crl{a,b,c}$ and four types of online nodes $V=[4]$.
Type $1$ matches to $a$ only, $2$ to $a$ and $b$, $3$ to $c$ only, and $4$ to $b$ and $c$.
Observe that the only types inducing error are $2$ and $4$.

Assume the arrival at $t=3$ is $\JT{3}=2$.
Matching it to $a$ is an error if arrivals are $\crl{1,1},\crl{1,3},\crl{1,4}$, so the disagreement is $p_1^2+2p_1p_3+2p_1p_4$.
Matching it to $b$ is an error if arrivals are $\crl{4,4},\crl{3,4}$ with disagreement $p_4^2+2p_4p_3$.
Now assume $p_4^2+2p_4p_3=p_1^2+2p_1p_3+2p_1p_4$, so the Bayes selector is indifferent and thus say it matches to $a$.

At $t=2$, there is only an error if $\JT{2}=4$, in which case matching it to $b$ has disagreement $p_2$ and matching it to $c$ disagreement $p_3$.
In conclusion to the Bayes selector pays $p_1^2+2p_1p_3+2p_1p_4$ in the first stage plus $\min\crl{p_2,p_3}$ in the second with probability $p_4$.

The strategy that matches at $t=3$ type 2 to $b$ has disagreement $p_4^2+2p_4p_3=p_1^2+2p_1p_3+2p_1p_4$, thus lower than the Bayes selector.
To see this, note that if we match to $b$ there is no error at $t=2$.

Finally, the equation $p_4^2+2p_4p_3=p_1^2+2p_1p_3+2p_1p_4$ is satisfied, e.g., with $p_1=p_4/2$ and $p_3=p_4/4$.
\end{proof}

\section{Additional Details from Numerical Experiments}\label[appendix]{appen:numerics}

Competitive is described as follows.
For a given horizon $T$, let $K_j\defeq\lceil p_jT\rceil$.
We create a bipartite graph $G=(U,V,E)$, where $U$ is the static side and $V$ the online side.
The static side contains $B_i$ copies of each resource $i$, hence $\abs{U} = \sum_{i\in [d]} B_i$.
The online side contains $K_j$ copies of each type $j$, hence $\abs{V}=\sum_{j\in [n]} K_j$.
The edge set $E$ is the natural construction where each copy of $j\in [n]$ has edges to all copies of $i\in [d]$ according to the adjacency matrix $A$.
The weight $w_e$ on edge $e=(u,v)$ is $r_{ij}$ if $u$ is a copy of $i$ and $v$ is a copy of $j$.
Finally, define the following matching LP on the graph $G$, where $\lambda_{iljk}$ stands for the $l$-th copy of $i$ and $k$-th copy of $j$
\[
\begin{array}{rrll}
(P)\, \max  & \multicolumn{3}{c}{\sum_{i\in [d]}\sum_{l\in [B_i]}\sum_{j\in [n]:A_{ij}=1}\sum_{k\in [K_j]} r_{ij} \lambda_{iljk}} \\
\text{s.t.}&  \sum_{j\in [n]}\sum_{k\in [K_j]}\lambda_{iljk} &\leq 1& \forall i\in [d], l\in [B_i]\\
& \sum_{i\in [d]}\sum_{l\in [B_i]} \lambda_{iljk} &\leq 1 &\forall j\in [n],k\in [K_j] \\
& \lambda &\geq 0,
\end{array}
\]

Let $\lambda^\star$ be a solution to this LP.
Whenever a type $j$ arrives, Competitive draws $k\in [K_j]$ uniformly at random, then takes a vertex $u=il$ incident to $v=jk$ with probability $\lambda^\star_{iljk}$ and if $u=il$ is not taken, it matches $v$ to $u$.
We note that the process of copying nodes is not superfluous since the analysis of Competitive heavily relies on the fact that the LP is in this form.

Marginal Allocation is described as follows.
Let $x$ be a solution of $(P_T)$ in \cref{eq:coupled_lp_matching}, i.e., of the fluid LP, and let $f_i:[T]\times \crl{0,\ldots,B_i} \to \Rp$ be some functions specified later.
When a type $j$ arrives at $t$ and the current budgets are $B^t\in\N^d$, Marginal Allocation uses $f_i(t,B^t_i)-f_i(t,B^t_i-1)$ as the bid-price for each resource $i\in [d]$: the type is rejected if $r_{ij}<f_i(t,B^t_i)-f_i(t,B^t_i-1)$ for all $i\in [d]$ such that $B^t_i>0$ and otherwise it is matched to $\argmax\crl{r_{ij}-f_i(t,B^t_i)+f_i(t,B^t_i-1):i\in[d],B^t_i>0}$.
Finally, the functions $f$ are obtained with the following recursion
\[
f_i(t+1,b)=f_i(t,b)+\frac{1}{T}\sum_{j\in [n]}x_{ij}(r_{ij}-f_i(t,b)+f_i(t,b-1))^+, \qquad 
f_i(1,\cdot) = 0, f_i(\cdot,0) = 0.
\]

\begin{table}[ht]
\small
\centering
\begin{tabular}{lllllllllllllllll}
                               &                         & \multicolumn{15}{c}{Type $j$}                                                                                   \\
                               & \multicolumn{1}{l|}{}   & 1     & 2     & 3     & 4     & 5    & 6     & 7     & 8   & 9   & 10   & 11    & 12    & 13    & 14    & 15    \\ \cline{2-17} 
\multirow{20}{*}{\rotatebox[origin=c]{90}{Resource $i$}} & \multicolumn{1}{l|}{1}  & 0     & 1     & 0     & 1     & 0    & 1     & 0     & 0   & 0   & 0    & 1     & 0     & 1     & 0     & 0     \\
                               & \multicolumn{1}{l|}{2}  & 1     & 1     & 1     & 0     & 0    & 1     & 1     & 1   & 1   & 0    & 0     & 0     & 1     & 1     & 1     \\
                               & \multicolumn{1}{l|}{3}  & 0     & 1     & 1     & 0     & 1    & 0     & 0     & 1   & 0   & 0    & 1     & 1     & 0     & 1     & 0     \\
                               & \multicolumn{1}{l|}{4}  & 0     & 1     & 0     & 0     & 0    & 1     & 1     & 0   & 1   & 1    & 1     & 1     & 1     & 0     & 0     \\
                               & \multicolumn{1}{l|}{5}  & 1     & 1     & 1     & 0     & 0    & 0     & 1     & 0   & 1   & 0    & 0     & 0     & 0     & 1     & 1     \\
                               & \multicolumn{1}{l|}{6}  & 0     & 1     & 1     & 0     & 1    & 0     & 1     & 0   & 1   & 0    & 1     & 1     & 1     & 1     & 0     \\
                               & \multicolumn{1}{l|}{7}  & 1     & 0     & 1     & 1     & 1    & 1     & 1     & 0   & 1   & 1    & 1     & 0     & 0     & 1     & 1     \\
                               & \multicolumn{1}{l|}{8}  & 1     & 1     & 1     & 0     & 0    & 1     & 0     & 0   & 1   & 0    & 0     & 0     & 1     & 0     & 0     \\
                               & \multicolumn{1}{l|}{9}  & 0     & 1     & 1     & 1     & 1    & 0     & 0     & 1   & 0   & 0    & 0     & 1     & 0     & 0     & 0     \\
                               & \multicolumn{1}{l|}{10} & 1     & 1     & 0     & 1     & 1    & 0     & 0     & 0   & 1   & 0    & 1     & 1     & 0     & 0     & 0     \\
                               & \multicolumn{1}{l|}{11} & 1     & 1     & 0     & 1     & 0    & 0     & 1     & 0   & 1   & 0    & 0     & 0     & 0     & 0     & 1     \\
                               & \multicolumn{1}{l|}{12} & 0     & 1     & 1     & 1     & 0    & 1     & 1     & 1   & 1   & 1    & 0     & 1     & 0     & 0     & 1     \\
                               & \multicolumn{1}{l|}{13} & 0     & 0     & 1     & 0     & 0    & 1     & 0     & 1   & 1   & 0    & 1     & 1     & 0     & 1     & 1     \\
                               & \multicolumn{1}{l|}{14} & 1     & 0     & 0     & 0     & 0    & 1     & 0     & 1   & 0   & 0    & 1     & 1     & 1     & 0     & 0     \\
                               & \multicolumn{1}{l|}{15} & 0     & 0     & 0     & 0     & 0    & 0     & 0     & 0   & 0   & 1    & 0     & 1     & 0     & 1     & 0     \\
                               & \multicolumn{1}{l|}{16} & 1     & 0     & 1     & 0     & 0    & 1     & 0     & 1   & 0   & 1    & 1     & 0     & 0     & 0     & 1     \\
                               & \multicolumn{1}{l|}{17} & 0     & 1     & 0     & 1     & 1    & 0     & 0     & 0   & 0   & 1    & 1     & 1     & 0     & 0     & 1     \\
                               & \multicolumn{1}{l|}{18} & 0     & 0     & 1     & 1     & 0    & 1     & 1     & 0   & 1   & 1    & 0     & 0     & 0     & 0     & 1     \\
                               & \multicolumn{1}{l|}{19} & 1     & 1     & 1     & 1     & 0    & 1     & 0     & 0   & 0   & 1    & 0     & 0     & 0     & 0     & 0     \\
                               & \multicolumn{1}{l|}{20} & 0     & 0     & 0     & 1     & 0    & 0     & 1     & 1   & 1   & 0    & 0     & 1     & 1     & 1     & 1     \\
                               & $p_j$                   & 0.075 & 0.075 & 0.125 & 0.025 & 0.05 & 0.062 & 0.062 & 0.1 & 0.1 & 0.05 & 0.125 & 0.012 & 0.075 & 0.062 & 0.002 \\
                               & $r_j$                   & 7     & 5     & 16    & 1     & 1    & 20    & 10    & 18  & 7   & 14   & 17    & 19    & 14    & 1     & 2    
\end{tabular}
\caption{Parameters used for the second Online Packing instance.
Coordinates $(i,j)$ represent consumption $A_{ij}$. 
}\label{tab:packing_two}
\end{table}

\begin{table}[ht]
\centering
\begin{tabular}{llllllllllll}
                              &                        & \multicolumn{10}{c}{Type $j$}                             \\
                              & \multicolumn{1}{l|}{}  & 1   & 2   & 3   & 4   & 5   & 6   & 7   & 8   & 9   & 10  \\ \cline{2-12} 
\multirow{6}{*}{\rotatebox[origin=c]{90}{Resource $i$}} & \multicolumn{1}{l|}{1} & 10  & 6   & 0   & 0   & 9   & 8   & 2   & 0   & 0   & 1   \\
                              & \multicolumn{1}{l|}{2} & 1   & 0   & 0   & 0   & 0   & 0   & 2   & 0   & 0   & 8   \\
                              & \multicolumn{1}{l|}{3} & 0   & 0   & 0   & 0   & 0   & 0   & 2   & 0   & 0   & 6   \\
                              & \multicolumn{1}{l|}{4} & 0   & 26  & 0   & 0   & 1   & 0   & 3   & 0   & 0   & 11  \\
                              & \multicolumn{1}{l|}{5} & 1   & 4   & 0   & 0   & 0   & 0   & 0   & 0   & 0   & 13  \\
                              & \multicolumn{1}{l|}{6} & 7   & 4   & 12  & 11  & 10  & 12  & 18  & 2   & 0   & 0   \\
                              & $p_j$                  & 0.1 & 0.1 & 0.1 & 0.1 & 0.1 & 0.1 & 0.1 & 0.1 & 0.1 & 0.1
\end{tabular}
\caption{Parameters used for the second Online Matching instance.
Coordinates $(i,j)$ represent the reward $r_{ij}$ and $r_{ij}=0$ implies that it is not possible to match $i$ to $j$.
}\label{tab:matching_two}
\end{table}
\end{APPENDICES}

\end{document}